\documentclass[11pt,letterpaper]{article}

\usepackage[utf8]{inputenc}
\usepackage[T1]{fontenc}

\usepackage[margin=1in]{geometry}

\usepackage[small]{caption}
\usepackage{subcaption}
\usepackage{comment}

\usepackage{amsmath}
\usepackage{amsfonts}
\usepackage{amssymb}
\usepackage{amsthm}
\usepackage{amstext}
\usepackage{thmtools}
\usepackage{thm-restate}
\usepackage{mathtools}
\usepackage{xspace}
\usepackage{enumerate}
\usepackage{hyperref}
\usepackage{cleveref}
\usepackage{tikz}
\usetikzlibrary{arrows}
\usetikzlibrary{fit}%
\usepackage{xcolor}
\usepackage{cite}
\usepackage{paralist}
\usepackage{csquotes}
\usepackage{bm}

\definecolor{col1}{HTML}{1AB6AB}
\definecolor{col2}{HTML}{EE8E16}

\tikzset{
asg/.style={black, line width=2pt, -stealth},
}

\newcommand{\gear}[5]{%
\foreach \i in {1,...,#1} {%
  [rotate=(\i-1)*360/#1]  (0:#2)  arc (0:#4:#2) {[rounded corners=1.5pt]
             -- (#4+#5:#3)  arc (#4+#5:360/#1-#5:#3)} --  (360/#1:#2)
}}  

\newcommand{\child}{
    \begin{tikzpicture}
        \draw (0,0) circle (10pt);
        \draw (-0.1,0.1) circle (1pt);
        \draw (0.1,0.1) circle (1pt);
        \draw (-0.175,-0.1) to [bend right=70] (0.175,-0.1);
    \end{tikzpicture}
}

\newcommand{\gift}{
    \begin{tikzpicture}
        \draw[black,scale=0.25] (-1,0) -- (1,0)
(-1,1) -- (1,1)
(-1,-1) -- (1,-1)
(-1,1) -- (-1,-1)
(0,1) -- (0,-1)
(0,1) -- (0,-1)
(1,1) -- (1,-1)
(0,1) -- (-0.5,1.3)
(0,1) -- (0.5,1.3);
    \end{tikzpicture}
}

\newcommand{\job}{
    \begin{tikzpicture}
        \draw[black,thick] (0,0) circle (5pt);
    \end{tikzpicture}
}

\newcommand{\machine}{
    \begin{tikzpicture}[scale=0.3]
        \draw[thick,black] (-1,-1) rectangle (1,1);
    \draw[scale=0.3] \gear{8}{2}{2.5}{10}{1};
    \draw (0,0) circle (5pt);
    \end{tikzpicture}
}

\newtheorem{theorem}{Theorem}
\newtheorem{lemma}[theorem]{Lemma}

\newtheorem{corollary}[theorem]{Corollary}

\newtheorem{definition}[theorem]{Definition}

\newcommand{\opt}{\mathrm{OPT}}
\newcommand{\optc}{\mathrm{OPT}_{\cC}}

\newcommand{\santa}{\textsc{SantaClaus}\xspace}

\newcommand{\matroidsanta}{resource-matroid \textsc{SantaClaus}\xspace}

\newcommand{\makespan}{\textsc{Makespan}\xspace}
\newcommand{\matroidmakespan}{job-matroid \textsc{Makespan}\xspace}

\newcommand{\cM}{\mathcal{M}}
\newcommand{\cI}{\mathcal{I}}
\newcommand{\cT}{\mathcal{T}}
\newcommand{\cO}{\mathcal{O}}
\newcommand{\cB}{\mathcal{B}}
\newcommand{\cC}{\mathcal{C}}

\newcommand{\si}{u}
\newcommand{\bi}{w}
\newcommand{\PM}{\mathcal{P}}
\newcommand{\dPM}{\overline{\mathcal{P}}}
\newcommand{\bv}{\bar{v}}

\newcommand{\bhat}[1]{\bm\hat{#1}}

\DeclarePairedDelimiter{\abs}{\lvert}{\rvert}
\DeclarePairedDelimiter{\ceil}{\lceil}{\rceil}

\usepackage[colorinlistoftodos,prependcaption,textsize=scriptsize]{todonotes}

\title{Santa Claus meets Makespan and Matroids: \\ Algorithms and Reductions\thanks{We thank Schloss Dagstuhl for hosting the Seminar 23061 on Scheduling  in February 2023 where we had fruitful discussions on this topic.}}
\author{{\'E}tienne Bamas\thanks{Post-Doctoral Fellow, ETH AI Center, Switzerland. etienne.bamas@inf.ethz.ch. This work was partially supported by the Swiss National Science Foundation project
entitled “Randomness in Problem Instances and Randomized Algorithms” [project
number 200021–184656/1].} \and Alexander Lindermayr\thanks{Faculty of Mathematics and Computer Science, University of Bremen, Germany. \{linderal,nmegow,jschloet\}@uni-bremen.de.} \and Nicole Megow\footnotemark[3] \and Lars Rohwedder\thanks{Maastricht University, Netherlands. l.rohwedder@maastrichtuniversity.nl. Supported by Dutch Research Council (NWO) project ``The Twilight Zone of Efficiency: Optimality of Quasi-Polynomial Time Algorithms'' [grant number OCEN.W.21.268].} \and Jens Schl{\"o}ter\footnotemark[3]}

\date{}

\begin{document}

\maketitle

\begin{abstract}
    In this paper we study the relation of two fundamental problems 
    in scheduling and fair allocation: makespan minimization
    on unrelated parallel machines and max-min fair allocation, also known as the Santa Claus problem. 
    For both of these problems the best approximation 
    factor is a notorious open question; more precisely, whether there is a better-than-2 approximation for the former problem and whether there is a constant approximation for the latter. 
    
    While the two problems are intuitively related and history has shown
    that techniques can often be transferred between them, no formal
    reductions are known. We first show that an affirmative answer to the open question for makespan minimization implies the same for the Santa Claus problem by reducing the latter problem to the former.
    We also prove that for problem instances with only two input values both questions are equivalent.
    
    We then move to a special case called ``restricted assignment'',
    which is well studied in both problems. Although our reductions
    do not maintain the characteristics of this special case, we give a reduction in a slight generalization, where the jobs or resources are assigned to multiple machines or players subject to a matroid constraint and in addition we have only two values. 
    Since for the Santa Claus problem with matroids the two value case is up to constants equivalent to the general case, this draws a similar picture as before: equivalence for two values and the general case of Santa Claus can only be easier than makespan minimization.
    To complete the picture, we give an algorithm for our
    new matroid variant of the Santa Claus problem
    using a non-trivial extension
    of the local search method from restricted assignment. Thereby we 
    unify, generalize, and improve several previous results.
    We believe that this matroid generalization may be of independent interest
    and provide several sample applications.

    As corollaries, we obtain a polynomial-time $(2 - 1/n^\epsilon)$-approximation for two-value makespan minimization for every $\epsilon > 0$, improving on the previous $(2 - 1/m)$-approximation, and a polynomial-time $(1.75+\epsilon)$-approximation for makespan minimization in the restricted assignment case with two values, improving the previous best rate of $1 + 2 / \sqrt{5} + \epsilon \approx 1.8945$.
\end{abstract}

\thispagestyle{empty} 
\addtocounter{page}{-1}
\newpage

\section{Introduction} 

In this paper we study two prominent topics from scheduling theory: the \santa problem and unrelated-machine \makespan minimization; in particular,
two notoriously difficult questions about polynomial-time approximations that are considered major open problems in the field~\cite{schuurman1999polynomial,Woeginger02,Bansal17,WilliamsonShmoys-book}.

In the \santa problem (also known as max-min fair allocation), we are given a set of~$m$ players~$P$ and a set of $n$ indivisible resources $R$. Each resource $j \in R$ has unrelated values $v_{ij} \geq 0$  for each player $i \in P$. 
The task is to find an assignment of resources to players with the objective to maximize the minimum total value assigned to any player. This objective is arguably the best from the perspective of {\em fairness} for each individual player. The name \enquote{Santa Claus} is due to Bansal and Sviridenko~\cite{bansal2006santa} who stated this problem as Santa's task to distribute gifts to children in a way that makes the least happy child maximally happy. 
From the perspective of approximation algorithms, it is entirely plausible that there 
exists a polynomial-time constant approximation for the problem, with the best lower bound assuming P$\neq$NP being only~$2$~\cite{bezakova2005allocating}.
On the other hand, the state-of-the-art is ``only'' a polynomial-time $n^{\epsilon}$-approximation for any constant $\epsilon > 0$,  a remarkable 
result by Chrakabarty, Chuzhoy, and Khanna~\cite{chakrabarty2009allocating},
who also give a polylogarithmic approximation in quasi-polynomial time using the same approach. Positive evidence towards a constant approximation comes from an intensively studied special case, called \emph{restricted assignment}. Here, the values satisfy $v_{ij}\in\{0, v_j\}$, or equivalently each resource has one fixed value, but can only be assigned to a specific subset of players. The inapproximability from the general case still holds for restricted \santa,
even for instances with only two non-zero values, introduced formally later.
The first constant approximations have been achieved
first by randomized rounding of the so-called configuration LP combined with Lov\'asz Local Lemma (LLL)~\cite{bansal2006santa, feige2008allocations} and later using a sophisticated local search technique analyzed against the configuration LP~\cite{annamalai2017combinatorial,DBLP:journals/talg/AsadpourFS12,DBLP:conf/icalp/ChengM18,DBLP:conf/icalp/ChengM19,davies2020tale,HaxellS2023improved,PolacekS2015}. The local search method originates from work on hypergraph matchings by Haxell~\cite{haxell1995condition}.

The \enquote{dual} problem with a min-max objective is the equally fundamental and prominent problem of scheduling jobs on unrelated parallel machines so as to
minimize the maximum completion time, that is, the makespan. For brevity we refer to this as the \makespan problem. 
Formally, we are given a set of $m$ machines $M$ and a set of $n$ jobs $J$. Every job $j \in J$ has size (processing time)~$p_{ij} \geq 0$ on machine $i \in M$. The task is to find an assignment of jobs to machines that minimizes the maximum load over all machines. Here, the load of a machine is the total size of jobs assigned to that machine.
Lenstra, Shmoys and Tardos~\cite{lenstra1990approximation} gave a beautiful $2$-approximation algorithm based on rounding a sparse vertex solution of the so-called assignment LP (a simpler relaxation than the configuration LP). The rounding has been slightly improved to the factor $2-1/m$~\cite{ShchepinV05}, but despite substantial research efforts, this upper bound remains undefeated.
The best lower bound on the approximability assuming P$\neq$NP is $3/2$~\cite{lenstra1990approximation}. Similar to the \santa problem, the restricted assignment case  with $p_{ij}\in\{p_j,\infty\}$ has also been studied extensively here. However, the barrier of~$2$ has been overcome only partially: 
with non-constructive integrality gap bounds~\cite{svensson2012santa} and better-than-$2$ approximations in quasi-polynomial time~\cite{jansen2020quasi} and for the special case of two sizes~\cite{chakrabarty20141,annamalai2019lazy}.

Intuitively, the two problems are related and in the community the belief has been mentioned that \santa admits a constant approximation if (and only if) \makespan admits a better-than-2 approximation; see e.g.,~\cite{Bansal17,BamasR2023-trees}.
Indeed, techniques for one problem often seem to apply to the other one, but no formal reductions are known. 
We give some examples for such parallels: 
\begin{compactenum}
\item[$(i)$] The configuration LP, see~\cite{bansal2006santa}, is the basis of all mentioned results for the restricted assignment variant of both problems.
\item[$(ii)$] The local search technique by Haxell~\cite{haxell1995condition} for hypergraph matching has been adopted and shown to  be very powerful for both problems. First, it has been picked up for restricted \santa \cite{DBLP:journals/talg/AsadpourFS12,PolacekS2015,annamalai2017combinatorial,davies2020tale} and later 
it has been further developed for restricted \makespan~\cite{svensson2012santa,jansen2020quasi,annamalai2019lazy}. 
\item[$(iii)$] Chakrabarty, Khanna and Li~\cite{chakrabarty20141} transferred the technique of rounding the configuration LP via LLL used for restricted \santa~\cite{feige2008allocations,HaeuplerSS11} to restricted \makespan with two job sizes and thereby provided the first slightly-better-than-two approximation in polynomial time. 
\item[$(iv)$] 
 The reduction for establishing hardness of approximation less than~$2$ 
 for \santa~\cite{bezakova2005allocating} is essentially the same as the earlier construction for the $3/2$-inapproximability for \makespan~\cite{lenstra1990approximation}.
 \item[$(v)$] The LP rounding by Lenstra et al.~\cite{lenstra1990approximation} achieves an additive approximation within the maximum (finite) processing time $p_{\max}$, i.e., the makespan is at most $\opt+p_{\max}$, which can be translated into a multiplicative $2$-approximation. 
 Bezakova and Dani~\cite{bezakova2005allocating} show the same additive approximation for \santa: each player is guaranteed a value at least $\opt - v_{\max}$. Note that for the max-min objective this does not translate into a multiplicative 
 guarantee. 
\end{compactenum}

\smallskip 
In this paper, we confirm (part of) the conjectured relation between the \makespan and the \santa problem with respect to their approximability. As our first main result we prove that a better-than-$2$ approximation for \makespan implies an $\cO(1)$-approximation for \santa.
\begin{theorem}\label{th:main-general}
    For any $\alpha \ge 2$, if there exists a polynomial-time $(2-1/\alpha)$-approximation for \makespan,
    then there exists a polynomial-time $(\alpha+\epsilon)$-approximation for \santa for any $\epsilon>0$.
\end{theorem} 
For values of $\alpha < 2$, a $(2 - 1/\alpha)$-approximation for the \makespan problem is NP-complete and the implication would still hold, even though clearly uninteresting. Similarly to this, we will restrict our attention in later theorems to non-trivial values of $\alpha$.

We prove also the reverse direction for the {\em two-value} case, in which all resource values are $v_{ij}\in\{0, \si, \bi\}$, for some $\si, \bi \ge 0$, and all processing times are $p_{ij}\in\{\si, \bi, \infty\}$, for some $\si, \bi \ge 0$, respectively. 
This implies the equivalence of \makespan and \santa in that case.
\begin{restatable}{theorem}{twoValueEquivalence}\label{thm:2value-equivalence}\label{th:main-twovalue}
For any $\alpha \geq 2$, there exists an $\alpha$-approximation algorithm for two-value \santa if and only if there exists a $(2-1 / \alpha)$-approximation algorithm for two-value \makespan.
\end{restatable}
We then move to the restricted assignment case, where one might hope to unify previous results and possibly infer new results by showing a similar relationship.
Using our techniques, however, this seems unclear,
since the aforementioned reductions do not maintain the characteristics of the restricted assignment
case. However, it turns out to be useful to consider a matroid generalization of our problems.
Towards this, we briefly introduce the notion of a matroid.
A {\em matroid} is a non-empty, downward-closed set system $(E, \cI)$ with ground set~$E$ and a family of subsets $\cI\subseteq 2^{E}$, which satisfies the {\em augmentation property}:
\begin{equation}\label{ex1} 
\mbox{if } I,J \in \cI \mbox{ and } |I|<|J|, \mbox{ then } I+j \in  \mathcal{I} \mbox { for some } j\in J\setminus{I}.
\end{equation}
Given a matroid $\cM=(E,\cI)$, a set $I\subseteq E$
is called \emph{independent} if $I\in \cI$, and \emph{dependent} otherwise. 
An inclusion-wise maximal independent subset is called \emph{basis} of~$\cM$, and we denote the set of bases by $\cB(\cM)$. With a matroid $\cM$, we associate a rank function $r: 2^{E} \rightarrow \mathbb Z_{\ge 0}$, where $r(X)$ describes the maximal cardinality
of an independent subset of $X$. Typical examples of matroids include: linearly independent subsets of some given vectors, acyclic edge sets of a given undirected graph, and subsets of cardinality bounded by some given value. A \emph{polymatroid} is the multiset
analogue of a matroid. We refer to Section~\ref{sec:notation} for further definitions and to~\cite{schrijver2003combinatorial} for a general introduction~to~matroids.

Moving back to our two problems, we first introduce
the restricted {\em \matroidsanta} problem, where we consider again an input of resources and players and each resource $j$
has a value $v_{j}$ for each player $i$ as in the restricted assignment case.
In the matroid variant, however, each resource can potentially be assigned to multiple players, subject to
a (poly-)matroid constraint; more precisely, we require the set of players, which we assign the resource to, to be a basis of a given resource-specific (poly-)matroid, and the resource  contributes to the total value of each of these players. As one may also consider a more general variant with unrelated values $v_{ij}$ we use the phrase {\em restricted} to emphasize our model.
Similarly, we define a restricted {\em \matroidmakespan} problem by replacing the max-min with a min-max objective and asking for an assignment of each job to a set of machines which forms a basis in the job's matroid. 

Davies, Rothvoss and Zhang~\cite{davies2020tale} recently introduced a closely related matroid variant of restricted \santa. Their variant, however, is significantly more restrictive and can be
summarized as follows: they allow a single infinite value matroid-resource and
a set of ``small value'' traditional resources (without the matroid generalization).
For this variant they give a $(4 + \epsilon)$-approximation algorithm.

First, we show that in our general variant an analogous relation to Theorem~\ref{th:main-twovalue} holds. 
\begin{restatable}{theorem}{ThmReductionRestrictedMatroidSantaTwoValue}

\label{thm:restricted-matroid-equivalence}
For any $\alpha \geq 2$, there exists a polynomial-time $\alpha$-approximation algorithm for the restricted two-value \matroidsanta problem if and only if there exists a polynomial-time $(2-1/\alpha)$-approximation algorithm for the restricted two-value \matroidmakespan problem.
\end{restatable}

Since the matroid version is a new problem, we cannot directly infer any approximation results for it. 
Hence, we develop a polynomial-time approximation algorithm for the restricted \matroidsanta problem.
\begin{theorem}\label{th:main-approx}
    For any $\epsilon > 0$, there exists a polynomial-time $(8+\epsilon)$-approximation algorithm for the restricted \matroidsanta problem and a $(4 + \epsilon)$-approximation in the two-value~case.
\end{theorem}
This is achieved via a non-trivial generalization of the commonly used local search technique for restricted assignment, see Section~\ref{sec:algorithmic-techniques} for a technical overview.
We prove the variant for two values and then give a reduction from the general case, see Lemma~\ref{lemma:two-value-santa-to-general-matroid}.
Apart from the curiosity-driven motivation for a matroid generalization of the classical  
scheduling problems and the usage through our reductions,
we present in Section~\ref{sec:applications} sample applications where such a variant arises.
Next, we state two immediate implications of our results to the state-of-the-art
of the \makespan problem.

\begin{corollary}\label{cor:1.75}
For every $\epsilon > 0$, there exists a polynomial-time $(1.75+\epsilon)$-approximation algorithm for the restricted two-value \matroidmakespan problem.
\end{corollary}
For completeness, we also provide a $2$-approximation of the restricted \matroidmakespan problem (with any number of values) that follows from standard techniques, see Theorem~\ref{thm:rounding2}.
Corollary~\ref{cor:1.75} holds in particular true for the restricted \makespan problem, thus
improving upon the previously best polynomial-time approximation rate of $1 + 2/\sqrt{5} + \epsilon \approx 1.8945$~\cite{annamalai2019lazy}.
The corollary follows from combining Theorems~\ref{thm:restricted-matroid-equivalence} and~\ref{th:main-approx}. In their work on restricted \santa, Davies et al.~\cite{davies2020tale} managed to reduce the technical complexity of previous works, which handled complicated path decompositions explicitly, using a cleaner matroid abstraction.
Our algorithm shows that such a simplification is also possible for
restricted two-value \makespan, which was not clear before.

\begin{corollary}
For every $\epsilon > 0$, there exists a polynomial-time $(2 - 1 / n^{\epsilon})$-approximation algorithm and a quasi-polynomial-time $(2 - 1/\mathrm{polylog}(n))$-approximation algorithm for two-value \makespan.
\end{corollary}
This result follows from the algorithm of Chakrabarty et al.~\cite{chakrabarty2009allocating} and Theorem~\ref{th:main-twovalue}, and
improves upon the best-known polynomial-time approximation factor of $2 - 1/m$ for $m$ machines~\cite{ShchepinV05}.
\pagebreak
\subsection{Applications}\label{sec:applications}
Next, we lay out three sample applications of our matroid generalization.

First, consider service centers that offer various types of services to clients. The specific service that such a center offers has some value associated with it and it can only be provided to a limited number of clients, a typical constraint appearing for example in capacitated facility location problems.
Furthermore, a service center can serve only clients that are located in the same region and a client can only receive a specific type of service once, i.e., by a single center, since receiving the same service twice yields no additional value.
The services should be assigned to the clients in such a way that all clients are treated \enquote{fairly} with respect to their total value for the services. That is, we want to maximize the total value for the least happy client.  This can be modeled as a resource-matroid \santa problem with clients being players and services (one per service type) being resources. The set of clients that can receive a particular type of service can be modeled as a transversal matroid\footnote{Given a bipartite graph $G=(J \cup S, E)$, a set $S' \subseteq S$ is independent in the {\em transversal matroid} $\cM=(S, \cI)$ if there is a matching in $G$ which covers $S'$.}.
In the classical (restricted) \santa problem, one cannot express the constraint that a client can receive each type of service only once.

As a second example, consider a program committee (PC) for a scientific conference. We would like to assign papers to PC members such that the workload is balanced in the sense that we minimize the maximum workload over all PC members. 
We will view this as a \makespan problem with PC members being machines and submissions being jobs.
PC members have declared which submissions they would agree to assess. Submissions may be of different types such as ``short papers'' or ``regular paper'' with varying workloads.
In a typical conference, each submission needs to be assessed three times and obviously it is important that this is done by different PC members; hence, we cannot simply model this as a traditional restricted assignment problem where we duplicate a job three times. However, it is easy to model this using matroids by having one basis for each triple of PC members that agree to assess this submission, i.e., we have a job-matroid \makespan problem with a uniform matroid\footnote{A {\em uniform matroid} $\cM=(X,\cI)$ of rank $r$ has as independent sets all subsets of $X$ of cardinality at most~$r$.} of rank $3$.

Our third illustration is a job-matroid \makespan problem, in which a 
graphic matroid\footnote{Given an undirected graph $G = (V,E)$, the {\em graphic matroid} $\cM = (E, \cI)$ has as independent sets the cycle-free edge sets (forests), i.e., $\cI = \{ F \subseteq E: F \text{ is acyclic in } G\}.$} allows us to model connectivity requirements.
In cloud computing and data centers, a number of servers is available to execute multiple applications at the same time.  Each application is executed on a subset of servers and these servers must be connected to allow for communication. We assume that these connections are direct in the sense that an application may not use additional servers as Steiner nodes.
We need to reserve a certain bandwidth for each application's communication, which depends on characteristics of the application itself. The task is to select carefully on which links to reserve the bandwidth for each individual application such that load on these links is balanced, more precisely, we want to select links to minimize the maximum total bandwidth requirement imposed on any link. This can be modelled as job-matroid \makespan problem with jobs being applications, machines being the edges (links) in a graph formed by the allocated servers, and the load (processing time) being the requested bandwidth of the application. The task is to choose for each job a spanning tree, that is, a basis in the job-dependent graphic matroid such that the maximum total bandwidth on any edge is minimized. This cannot be modeled as classical \makespan or restricted assignment problem, since it cannot capture the structure of a graphic matroid.

\subsection{Algorithmic techniques}\label{sec:algorithmic-techniques}
Our main algorithmic contribution lies in a local search algorithm for the new variant restricted \matroidsanta, see Theorem~\ref{th:main-approx}.
We give an overview of the method here, its main technical merits, and how it relates to previous works.
The specific local search method that we refer to originates in an algorithmic proof for a hypergraph matching theorem by Haxell~\cite{haxell1995condition}.
The theorem is a generalization of Hall's theorem for bipartite graphs to hypergraphs
and Haxell's proof can be thought of as a very non-trivial extension of the augmenting path method in bipartite graphs.
Asadpour et al.~\cite{DBLP:journals/talg/AsadpourFS12} made the connection to restricted \santa. In addition to a black-box reduction to the specific hypergraph matching problem, they also reinterpreted Haxell's method as a sophisticated
algorithm for restricted \santa.

Although not explicitly mentioned in earlier works, this new algorithm can be
thought of as a generalization of the typical augmentation algorithm for matroid partition: the core of the problem lies in the case where we have two values for the resources, more precisely, either the value of a resource is infinitely large or it is a unit value~$1$. This case is up to constants equivalent to the general problem, see e.g.,~\cite{bansal2006santa}.
Observe now the following structure: we need to select a subset $I_M$ of players such that there exists a left-perfect matching of $I_M$ to the infinite-value resources
and such that there exists a $b$-matching of all players in $P\setminus I_M$ to the small resources (each player is matched to $b$ resources, each resource to at most one player), where $b$ has to be maximized.
The sets $I_M$ that fulfill the condition above form a transversal matroid, but unfortunately the sets of players for which there is a $b$-matching does not (for a fixed $b > 1$). If they \emph{would} actually form a matroid, then the problem could easily be solved by matroid partition, where given two matroids $\cM_1 = (E, \cI_1), \cM_2 = (E, \cI_2)$ over the same ground set, we want to find two independent sets $I_1\in\cI_1$, $I_2\in\cI_2$ that partition the ground set, i.e., $I_1 \dot\cup I_2 = E$ (here we focus on the variant with two matroids, although also more than two matroids may be allowed).
Matroid partition--and the equivalent problem of matroid intersection--can be solved
in polynomial time by an augmenting path algorithm that repeatedly increases the
union of $I_1$ and $I_2$ by first swapping elements between these two sets.
Although, as mentioned above, the $b$-matching does not have a matroid structure,
the algorithmic paradigm of swapping elements between matching and $b$-matching
in order to increase their union still works once we allow approximation of $b$.

The constraint on set $I_M$ forms a transversal matroid (implied by the infinite-value resources) and Davies et al.~\cite{davies2020tale} then showed that
the algorithmic idea generalizes to arbitrary matroid structures.
In their problem, however, the $b$-matching remains the same without further abstraction, whereas in our further generalization, we embrace the polymatroid structure of the $b$-matching.
We now require instead of a $b$-matching that the multiset $b\cdot (P\setminus I_M)$ (having $b$ copies for each element in $P\setminus I_M$) is in some given polymatroid.
We believe that this abstraction is the logical conclusion for this line of research.

Although a seemingly natural extension, it is highly non-trivial to generalize the
existing algorithm to our setting. Firstly, there are conceptional issues that come from the fact that previous methods revolve around reassigning resources or jobs and those do not exist explicitly in our polymatroid. 
Secondly, a serious technical problem comes from the lack of a certificate of infeasibility. The design of the algorithm is closely connected to a certificate of infeasibility, which is provided (for analysis' sake) in case the
algorithm fails. For example, in matroid partition when the augmenting path algorithm
fails, one can derive a set $X$ such that $r_1(X) + r_2(X) < |X|$, where $r_1$ and $r_2$ are the rank functions of the matroids~\cite{knuth1973matroid}. This clearly proves infeasibility.
In applications of the method to restricted \santa or \makespan,
the role of the certificate of infeasibility was taken by the configuration LP.
It is unclear how one would generalize the configuration LP beyond the partial matroid generalization of Davies et al.~\cite{davies2020tale},
since it heavily relies on the matching structure of small resources. But even worse, we show that already in a special case for which
the configuration LP is still meaningful, its integrality gap is large; hence it is not helpful. Consider the following instance (which appeared already in \cite{bansal2006santa}). We have a first set $E$ of $m$ players, and a set $S$ of $m-1$ resources. For each player $i\in E$, there is an additional set $E_{i}$ of $m$ players and an additional set $S_{i}$ of $m+1$ resources. Each player $i\in E$ has valuation $m$ for all the resources in $S$, valuation $1$ for all the resources in $E_{i}$, and valuation $0$ for the remaining resources. For any $i\in E$, each player in $E_{i}$ has valuation $m$ for the resources in $S_{i}$, and valuation $0$ for other resources. It was showed in \cite{bansal2006santa} that in this example the configuration LP will give of $m$ to this instance, while it is clear that the integral optimum is at most $1$. Indeed, there are only $m-1$ resources in $S$, hence at least one player $i$ from $E$ will need to take resources from $S_{i}$. But $S_{i}$ contains only $m+1$ resources while $m$ players in $E_{i}$ need to take one resource each from the corresponding set $S_{i}$. Interestingly, this example can be captured by our polymatroid variant. We have universe $E$ and the uniform matroid of rank $|E|-1$ (i.e. $r(X)=|X|$ for any $X\neq E$, and $r(E)=|E|-1$), and $f(X)=|X|$, which models that a set of $|X|$ players in $E_1$ can be assigned a maximum of $|X|$ resources from the set $\bigcup_{i\in E} S_{i}$. In the matroid problem, the players in $\bigcup_{i\in E} E_{i}$
only appear implicitly. One may also derive $f$ by defining a natural polymatroid that assigns the items in $\bigcup_{i\in E} S_{i}$ and from which we then contract the players of $\bigcup_{i\in E} E_i$. As stated above, the configuration LP does not give us a good lower bound (or certificate of infeasibility) in this case.
Another way to find a certificate would be, in the spirit of matroid partition, to try to find a set $X$, for which $r(X) + f(X) / b < |X|$, where $r$ is the rank function of the matroid and $f$ the submodular function corresponding to the polymatroid. Although this would prove infeasibility of a solution of value $b$, it is not sufficient as seen in the example above. As previously detailed, it is clear that no solution of value more than $1$ exists. On the other hand, $r(X) + f(E)/|E| \ge r(X) = |X|$ for all $X \neq E$ and $r(E) + f(E)/|E| = |E| - 1 + 1 \ge |E|$. Hence, the certificate above is also not sufficient to rule out a solution of value $|E| = m$.

To overcome this issue, we develop the following certificate of infeasibility. Let $X\subseteq E$ with $f(X) \le b\cdot |X|$, in the example above we can take $X = \emptyset$. Further, let $Y\supseteq X$ have a not too large rank of $r(Y) < |Y| - |X| / 2$ and small marginal values for each element, that is, 
$f(i\mid X) < b$ for all $i\in Y\setminus X$. In the example above, take $Y = E$.
We claim that this constitutes a proof that no solution of value $3b$ exists.
Suppose that for $I_P\subseteq Y$ we have that $3b\cdot I_P$ is in the polymatroid. Then
\begin{equation*}
    2 b \cdot |I_P| \le 3b \cdot |I_P| - \sum_{i\in I_P\setminus X} f(i \mid X) \le f(X) \le b\cdot |X| \ ,
\end{equation*}
where we use that $3b \cdot |I_P| \le f(I_P) \le f(X) + \sum_{i\in I_P\setminus X} f(i\mid X)$.
It follows that $|I_P| \le |X| / 2$, but because of its
rank the matroid cannot cover all remaining elements.
Our algorithm is carefully designed to either output
a solution of value $b$ or to prove, using the idea above,
that no solution exists for $(4 + \epsilon)b$. Note
that we use a slightly higher constant compared to the
miniature above and we also require a more complicated variant of the certificate. This is mainly for efficiency reasons: in
order to achieve polynomial running time we need to work with weaker conditions.

\subsection{Definitions and notation} \label{sec:notation}

We write $\cO_\epsilon(\cdot)$ as the standard $\cO$-notation, where we suppress any factors that are functions in only $\epsilon$.
For a set $X$ and an element $i$ we write $X + i := X\cup\{i\}$. Similarly, $X - i := X \setminus \{i\}$. 

Let $E$ be a universe.
For a vector $x \in \mathbb{R}^{E}$, we write $x(e)$ for the entry of $x$ corresponding to $e \in E$, and $x(S) = \sum_{e \in S} x(e)$.
For some $X\subseteq E$, we write $b \cdot X$ as the vector $y \in \mathbb Z^E$ with
$y(e) = b$ for $e\in X$ and $y(e) = 0$ for $e\notin X$.
A set function $f : 2^{E} \to \mathbb{R}$ is submodular if for all subsets $A,B \subseteq E$ holds $f(A) + f(B) \geq f(A \cup B) + f(A \cap B)$, and monotone if for all $A \subseteq B \subseteq E$ holds $f(A) \leq f(B)$.
Let $f: 2^{E} \to \mathbb{Z}_{\geq 0}$ be a monotone submodular integer function with $f(\emptyset) = 0$. 
An \emph{integer polymatroid} over $E$ associated with $f$ is defined as
\[
 \PM = \{x \in \mathbb{Z}_{\geq 0}^{E} : x(S) \leq f(S) \; \forall S \subseteq E \}.
\]
In the following we always refer to integer polymatroids when talking about polymatroids.
A polymatroid can be interpreted as the multiset generalization of a matroid and most concepts of matroids translate easily to polymatroids. Every element $x \in \PM$ can be seen as an independent multiset in which an element $e \in E$ appears with multiplicity $x(e)$. A polymatroid is also downward-closed, that is, $x \in \PM$ implies $y \in \PM$ for any $0 \leq y \leq x$, and satisfies the augmentation property, that is, if $x,y \in \PM$ with $x(E) < y(E)$, then there is some $e \in E$ such that $x' \in \PM$ with $x'(e) = x(e) + 1$ and $x'(e') = x(e')$ for all $e' \in E - e$.
In particular, any matroid is a polymatroid.

A basis of a polymatroid $\PM$ is an element $x \in \PM$ which satisfies $f(E) = x(E)$, meaning that all bases have the same cardinality (in terms of multisets). We denote the set of bases of $\PM$ by $\cB(\PM)$.
For a given polymatroid $\PM$ and a constant $k \in \mathbb{Z}_{\geq 0}$, the set $\{x \in \PM : x(e) \leq k \; \forall e \in E \}$ is again a polymatroid.

For a given polymatroid $\PM$ of the submodular function $f$, and some vector $z \in \mathbb{Z}_{\geq 0}^E$ with $x \leq z$ for all $x \in \PM$, we define the \emph{dual polymatroid} $\dPM$ of $\PM$ with respect to $z$ 
via the set function $g$ with
\[
g(S) = z(S) + f(E \setminus S) - f(E)
\]
for every $S \subseteq E$.
This function is submodular, monotone and satisfies $g(\emptyset)=0$, hence this definition is well-defined. Note that if $x \in \cB(\PM)$ it follows $g(E) = z(E) + f(\emptyset) - f(E) = z(E) - x(E)$, and therefore $z - x \in \cB(\dPM)$.

If a polymatroid $\PM$ associated with a function $f$ is given as an input for a problem, we assume that it is represented in form of a value oracle for $f$. We can test whether some vector $x$ is in $\PM$ by checking whether the minimum of the submodular function $f(S) - x(S)$ is non-negative, which can be done with a polynomial number of value queries to $f$.

We refer for an extensive overview over polymatroids to Schrijver~\cite[chapters 44 - 49]{schrijver2003combinatorial}.
We now give precise definitions of the matroid problems we consider. 

\begin{definition}[Resource-matroid \santa]
In the restricted {\em resource-matroid} \santa problem, there are sets of $m$ players $P$ and $n$ resources $R$ 
with values $v_{j}$ for all $j \in R$.
Further, for every resource $j \in R$ there is an integer polymatroid $\PM_j$ over $P$.
The task is to allocate each resource $j\in R$ to a basis $x_j\in \cB(\PM_j)$ and let each player $i$ profit from the resource $j$ with value $v_{j} \cdot x_j(i)$.  The goal is to maximize the minimum total value any player receives, i.e., 
\(
    \min_{i\in P} \sum_{j \in R} v_{j} \cdot x_j(i)  \,.
\)
\end{definition}

\begin{definition}[Job-matroid \makespan]
In the restricted {\em job-matroid} \makespan problem, there are sets of $m$ machines $M$ and $n$ jobs $J$ with sizes~$p_{j}$ for all $j \in J$. Further, for every job $j \in J$ there is an integer polymatroid $\PM_j$ over $M$. The task is to allocate each job $j \in J$ to a basis $x_j \in \cB(\PM_j)$ which means that $j$ contributes load $p_{j} \cdot x_j(i) $ to the total load of machine $i$. The goal is to minimize the maximum total load over all machines, i.e., 
\(
    \max_{i\in M} \sum_{j \in J} p_{j} \cdot x_j(i) \, .
\)
\end{definition}

These new matroid allocation problems generalize the restricted assignment variants of \santa and \makespan, respectively. In fact, the matroid variant with a uniform matroid 
of rank~$1$ corresponds to the respective traditional problem. 

Note that in restricted \matroidsanta it is equivalent to require that $x_j \in \PM_j$ for each resource $j$ instead of $x_j \in \cB(\PM_j)$, since we can always increase $x_j$ to reach a basis without making the solution worse. In restricted \matroidmakespan this is not the case.

Both matroid problems can be reduced to instances where the number of polymatroids is equal to the number of distinct job sizes (resource values). This is because we can sum polymatroids associated with jobs (resources) of equal size (value) to a single one, 
and then decompose a basis for such a merged polymatroid  
into bases for the original polymatroids via polymatroid intersection.
Formally, we get the following proposition. For a more detailed explanation, see~\Cref{apx:notation}. 

\begin{restatable}{proposition}{propreducejobs}\label{prop:matroid-makespan-reduce-jobs}
For any $\alpha \geq 1$, if there exists a polynomial-time $\alpha$-approximation algorithm for restricted \matroidmakespan (\matroidsanta) with $h$ jobs (resources), then there exists a polynomial-time $\alpha$-approximation algorithm for restricted \matroidmakespan (\matroidsanta) with $p_j$ resp. $v_j \in \{\bi_1,\ldots,\bi_h\}$ and $\bi_1,\ldots,\bi_h \geq 0$.
\end{restatable}

\section{Santa Claus and makespan reductions}
\label{sec:santa-to-makespan}

In this section we present our first two reductions and prove \Cref{th:main-general} and \Cref{th:main-twovalue}. 
The precise statements given in the following subsections imply these results. They are formulated as subroutines for a standard guessing framework (see e.g.\ \cite{HochbaumS87}), which we briefly explain here. 
Consider a \santa instance $I$ for which we want to compute an $\alpha$-approximate solution. We first guess $\opt(I)$ with some variable $T$ using binary search as follows.
For some guess $T$, we scale down all values of $I$ by factor $T$ and obtain $I'$. 
Then, we prove that if $\opt(I) \geq T$ (and $\opt(I') \geq 1$), our subroutine finds a solution for $I'$ with an objective value of at least $1/\alpha$. 
If we do not obtain such a solution, we can conclude $\opt(I) < T$ and safely repeat with a smaller guess. Otherwise, we repeat with a larger guess.
After establishing $T = \opt(I)$, the subroutine gives us a solution with an objective value of at least $T/\alpha$. For \makespan one can design an analogous procedure. 

\subsection{From Santa Claus to makespan minimization}

In this section, we present an approximation preserving reduction (up to a factor of $1+\epsilon$) from \santa to \makespan. More precisely, we show the following result. 

\begin{lemma}\label{thm:santa-to-makespan} 
For any $\alpha \geq 2$ and $\epsilon > 0$, given an instance $I$ of \santa with $\opt(I) \geq 1$, we can construct in polynomial time an instance $I'$ of \makespan such that, given a $(2-1/\alpha)$-approximate solution for $I'$, we can compute in polynomial time a solution for~$I$ with an objective value of at least $1/(\alpha+\epsilon)$.
\end{lemma}

This lemma then implies \Cref{th:main-general} via the guessing framework.
We split the proof of \Cref{thm:santa-to-makespan} into two parts (cf.~\Cref{lemma:configuration-reduction,thm:reduction-polyconfig-santa-to-makespan}).
First, we show that we can define a polynomial number of \emph{configurations} for each player, which represent different options this player has. We show that there is a nearly optimal
solution, up to a factor of $1+\epsilon$, that only uses these configurations.

Second, we reduce this problem to \makespan without losing additional constants. That is, we present a reduction proving that a $(2-1/\alpha)$-approximation for \makespan implies an $\alpha$-approximation for \santa restricted to polynomially many configurations.
Intuitively, the fact that we can enumerate the list of possible optimal configurations per player enables us to create gadgets for every configuration in the constructed \makespan instance which we exploit when reducing solutions.

\subsubsection{Santa Claus with polynomially many configurations}

Consider a \santa instance with players $P$ and $n$ resources~$R$.
We define the set of \emph{value types} as $\cT = \{v_{ij} : i \in P, j \in R\}$, which contains all distinct resource values that occur in the instance.
We call a function $c: \cT \to \{0,1,\ldots,n\}$ a \emph{configuration}, and define the \emph{total value} of $c$ as $\abs{c} = \sum_{v \in \cT} c(v) \cdot v$. One can also see a configuration as a multiset of value types.

Given a configuration $c_i$ for a player $i$ of a \santa instance $I$, we say that a resource assignment $A=\{A_i\}_{i\in P}$  for $I$ that assigns the set of resources $A_i$ to player $i$ \emph{matches} the configuration $c_i$ if  $\abs{\{ j \in A_i : v_{ij} = v \}} = c_i(v)$ for every value type $v \in \cT$. 

We use $\cC_i$ to refer to a set of configurations for a player $i \in P$ and call $\cC = \{\cC_i\}_{i \in P}$ a \emph{collection of configurations}.
A resource assignment $A$ matches a collection of configurations $\cC$ if, for each player $i$, there exists a configuration $c \in \cC_i$ such that $A_i$ matches $c$. 
Given a \santa instance $I$ and a collection of configurations $\cC = \{\cC_i\}_{i \in P}$, we use $\optc(I)$ to refer to the optimal objective value for instance $I$ among those solutions that match $\cC$.

The main result of this section is the following lemma.

\begin{restatable}{lemma}{santaConfigurationReduction}\label{lemma:configuration-reduction}
    For every $\epsilon > 0$ and a given instance $I$ of \santa with $\opt(I) \geq 1$, we can construct a rounded instance $I'$ with a collection of configurations $\cC$ such that the number of configurations for each player is polynomial in the input size of $I$ 
    and $\optc(I') \ge 1/(1+\epsilon)$. 
    
    Further, every solution for $I'$ of objective value $T$ is a solution for $I$ with objective value at least $T$.
\end{restatable}

The lemma essentially allows us to consider only solutions that partially match the constructed collection of configurations $\cC$. If we find such a solution that $\alpha$-approximates $\optc(I')$, we immediately get a $(\alpha+\epsilon)$-approximation for $\opt(I)$. 

To prove the lemma, we employ several rounding techniques and enforce a certain monotonicity condition on the configurations. This allows us to reduce the number of configurations per player to a polynomial while still guaranteeing $\optc(I') \ge 1/(1+\epsilon)$. For a full proof, we refer to~\Cref{apx:polynomial-configs}.

\subsubsection{Reduction to makespan minimization}\label{sec:csc-mm-red}

We prove the following lemma which, together with~\Cref{lemma:configuration-reduction}, implies~\Cref{thm:santa-to-makespan}.

\begin{lemma}\label{thm:reduction-polyconfig-santa-to-makespan}
Let $I$ be an instance of \santa and let $\cC$ be a collection of configurations with $\optc(I) \geq 1$. For any $\alpha \geq 1$, we can construct in polynomial time an instance $I'$ of \makespan such that, given a $(2-1 / \alpha)$-approximate solution for $I'$, we can compute in polynomial time a solution for $I$ with value at least $1 / \alpha$. 
The running times are polynomial in the size of $(I,\cC)$.
\end{lemma}

We first give some intuition for the reduction of the lemma.
Assume for simplicity that every configuration in~$\cC$ has value 1. We want to construct a \makespan instance with an optimal makespan equal to 1.
Exploiting the polynomial number of configurations, we introduce \emph{configuration-machines} for every player and every configuration of that player. By using a gadget structure, we ensure that every solution for the \makespan instance \enquote{selects} one configuration-machine for each player, which we call the \emph{player-machine}. 
This is done by forcing an extra (selection-)load of 1 on exactly one configuration-machine for that player. The other configuration-machines will just be able to absorb all other jobs that can be placed on the machine, so we assume that they do so and ignore them. 
Intuitively, the player-machine determines the configuration which we (partially) use for the player when transferring a solution back to the \santa instance.
To this end, we encode the corresponding configuration $c$ of the player-machine by introducing for every $v \in \cT$ a total of $c(v)$ \emph{configuration-jobs} with size $v$ on that machine.
For every resource we also introduce a \emph{resource-machine}, on which only configuration-jobs of the same value type as the corresponding resource can be placed, with size 1.
In an optimal solution of makespan 1, no configuration-job can be placed on a player-machine due to the selection-load.
This means that all these jobs have to be placed on the resource-machines instead. Since each resource-machine can absorb at most one job, we can interpret the placement of the configuration jobs as resource assignment for the \santa instance that matches the configuration of the player-machines. This, however, only works for optimal solutions.

To better understand the connection between approximate solutions, imagine an initial state where all configuration-jobs are placed on their player-machines in the \makespan instance, and all resources are unassigned in the \santa instance.
This means that all player-machines have a total load of~$1 + \abs{c} = 2$ (observe that $\abs{c}=1$ is caused by the configuration-jobs) and all players have a total value of $0$.
Now, a player can gain a resource by moving a suitable configuration-job away from her player-machine to a resource-machine. 
Since, even in a better-than-2 approximation for the \makespan instance, a resource-machine can absorb at most one job, we can again interpret this as the players competing for the resources via moving jobs away from their player-machines to resource-machines.
Therefore, in a $(2-1 / \alpha)$-approximation for the \makespan instance, a player must be able to move jobs away from her player-machine of total size at least $1 / \alpha$. But this means in our interpretation that she receives resources of total value at least $1 / \alpha$ in the \santa~instance.

In the following, we formalize these ideas and prove~\Cref{thm:reduction-polyconfig-santa-to-makespan}.
Fix an instance $I$ of the \santa problem and a collection of configurations $\cC$ for $I$ with $\optc(I) \geq 1$. We proceed by describing the reduction and proving two auxiliary lemmas that imply~\Cref{thm:reduction-polyconfig-santa-to-makespan}.

\paragraph{Preprocessing}
We first remove all configurations from $\cC$ of value strictly less than $1$. Since we assumed that $\optc(I) \geq 1$, this does not affect $\optc(I)$.

\paragraph{Construction}
We construct a \makespan instance~$I'$ as follows.
For every player $i$ in instance~$I$, we introduce a
player-job $j_i$ and for every configuration $c \in \cC_i$ a configuration-machine $m_c^i$, where the size of $j_i$ is equal to $1$ on every configuration-machine $m_c^i$ and $\infty$ on all other machines.
For every configuration $c \in \cC_i$ and for every value type $v \in \cT$ we introduce a set $J_{c,v}^i$ of $c(v)$ many configuration-jobs. 
For every resource $j$ we introduce a resource-machine $m_j$.
Finally, every configuration-job in $J_{c,v}^i$ has size $1$ on every resource-machine $m_j$ if resource $j$ has value type $v$ for player $i$ (i.e., $v_{ij} = v$), size $v / \abs{c}$ on the configuration-machine $m_c^i$, and $\infty$ on all other machines. See \Cref{fig:santa-to-makespan-construction}.

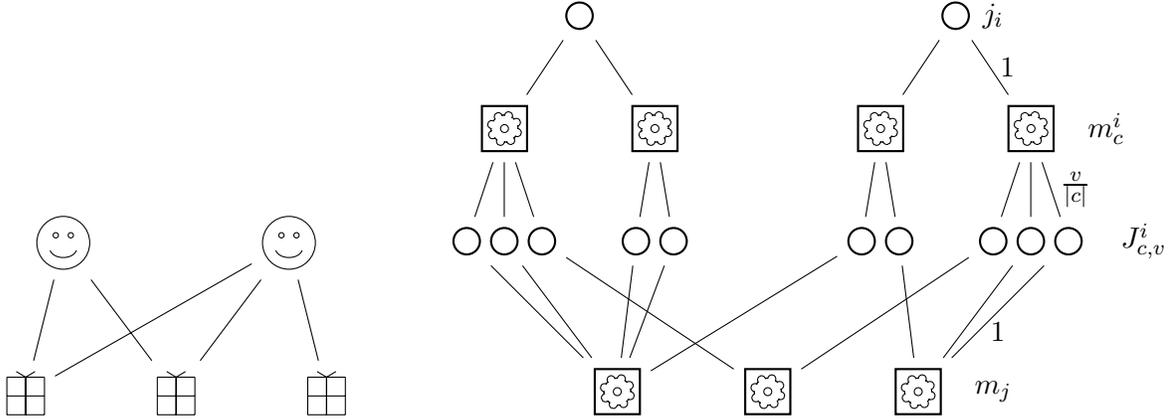
\begin{figure}
    \begin{subfigure}[b]{0.38\textwidth}
    \begin{center}
        \begin{tikzpicture}
            \node (c1) at (-1.5,5) {\child};
            \node (c2) at (1.5,5) {\child};
            \node (g1) at (-2,3) {\gift};
            \node (g2) at (0,3) {\gift};
            \node (g3) at (2,3) {\gift};

            \draw[black]
                (c1) -- (g1)
                (c1) -- (g2)
                (c2) -- (g1)
                (c2) -- (g2)
                (c2) -- (g3);
        \end{tikzpicture}
    \end{center}
    \caption{\santa instance $I$. Players are visualized by smileys and resources by gifts.}
    \end{subfigure}
    \hfill
    \begin{subfigure}[b]{0.6\textwidth}
    \begin{center}
    \begin{tikzpicture}
    \node (cj1) at (0,9) {\job};
    \node (m1) at (-1,7.5) {\machine};
    \node (m2) at (1,7.5) {\machine};
    \node (m1j1) at (-1.5,6) {\job};
    \node (m1j2) at (-1,6) {\job};
    \node (m1j3) at (-0.5,6) {\job};
    \node (m2j1) at (0.75,6) {\job};
    \node (m2j3) at (1.25,6) {\job};
    \draw[black] (cj1) -- (m1);
    \draw[black] (cj1) -- (m2);
    \draw[black] (m1j1) -- (m1);
    \draw[black] (m1j2) -- (m1);
    \draw[black] (m1j3) -- (m1);
    \draw[black] (m2j1) -- (m2);
    \draw[black] (m2j3) -- (m2);

    \node (cj2) at (5,9) {\job};
    \node at (5.5,9) {$j_i$};
    \node (m3) at (4,7.5) {\machine};
    \node (m4) at (6,7.5) {\machine};
    \node at (7,7.5) {$m_c^i$};
    \node (m3j1) at (3.75,6) {\job};
    \node (m3j3) at (4.25,6) {\job};
    \node (m4j1) at (5.5,6) {\job};
    \node (m4j2) at (6,6) {\job};
    \node (m4j3) at (6.5,6) {\job};
    \node at (7.5,6) {$J_{c,v}^i$};
    \draw[black] (cj2) -- (m3);
    \draw[black] (cj2) -- node[right] {1} (m4);
    \draw[black] (m3j1) -- (m3);
    \draw[black] (m3j3) -- (m3);
    \draw[black] (m4j1) -- (m4);
    \draw[black] (m4j2) -- (m4);
    \draw[black] (m4j3) -- node[right] {$\frac{v}{\abs{c}}$} (m4);

    \node (mr1) at (0.5,4) {\machine};
    \node (mr2) at (2.5,4) {\machine};
    \node (mr3) at (4.5,4) {\machine};
    \node at (5.5,4) {$m_j$};

    \draw[black] (mr1) -- (m1j1);
    \draw[black] (mr1) -- (m1j2);
    \draw[black] (mr2) -- (m1j3);
    \draw[black] (mr1) -- (m2j3);
    \draw[black] (mr1) -- (m2j1);
    \draw[black] (mr1) -- (m3j1);
    \draw[black] (mr2) -- (m4j1);
    \draw[black] (mr3) -- (m4j2);
    \draw[black] (mr3) -- node[below] {$1$} (m4j3);
    \draw[black] (mr3) -- (m3j3);

    \end{tikzpicture}
    \end{center}
    \caption{\makespan instance $I'$. Machines are visualized by squares with a gear and jobs by cycles.}
    \end{subfigure}
    \caption{The construction used in \Cref{thm:reduction-polyconfig-santa-to-makespan}. In both pictures an edge indicates that an item has a non-trivial value for an entity.}
    \label{fig:santa-to-makespan-construction}
\end{figure}

\begin{lemma}
The optimal objective value of $I'$ is at most $1$.
\end{lemma}

\begin{proof} 
Fix a solution of $I$ that is optimal among the solutions that match $\cC$.  
Consider a player $i$ of instance $I$ and let $c \in \cC_i$ be the selected configuration for player $i$ in the given solution. Let $A_i$ be the set of resources assigned to player $i$. In the solution for $I'$, we assign job $j_i$ to machine $m_c^i$, giving it a load of~$1$.
Further, we assign the configuration-jobs $J_{c,v}^i$ of configuration $c$ to resource-machines $\{m_j : j \in A_i \}$ such that every resource-machine receives at most one job. Such an assignment must exist by the fact that configuration $c$ is matched by the fixed solution for $I$. For every configuration $c' \in \cC_i \setminus \{c\}$ we assign for all $v \in \cT$ every configuration-job in $J_{c',v}^i$ to machine $m_{c'}^i$, giving those a load of $\sum_{v \in \cT} c'(v) \frac{v}{\abs{c'}} = 1$. 
Since, in the given solution, every resource $j$ is assigned to at most one player $i$, and since we have assigned at most one configuration-job
to machine $m_j$, every resource-machine also has a load of at most 1.
Hence, the makespan of the constructed solution for $I'$ is at most $1$.
\end{proof}

\begin{lemma}\label{lemma:sc-to-makespan-approx}
For any $\alpha \geq 1$, given a solution for $I'$ with a makespan of at most~$2 - 1 / \alpha$, we can construct in polynomial time a solution for $I$ where every player receives a total value of at least~$1 / \alpha$.
\end{lemma}
\begin{proof}
Given a solution for $I'$ where every machine has a load of at most $2-1 / \alpha$, we construct a solution for $I$ as follows. 
Fix a player $i$ and assume that $j_i$ is assigned to machine~$m_c^i$. 

Let $J_i$ be the set of configuration-jobs of configuration $c$ of player $i$ which are \emph{not} assigned to $m_{c}^i$.
Thus, every job in $J_i$ is assigned to a resource-machine. Note that every resource-machine has at most one assigned job, because every job has size of at least $1$ on these machines.
Let $R_i$ be the set of resources for which the corresponding resource-machines receive a job of $J_i$. 
We assign the resources $R_i$ to player $i$ in the solution for $I$. 
The load contributed by configuration-jobs to machine $m_c^i$ is at most $1-1 / \alpha$, because job $j_i$ is also assigned to $m_c^i$ and has size $1$. This implies that the total size of jobs in $J_i$ for machine $m_c^i$ is at least 
\[
\sum_{j \in J_i : j \in J_{c,v}^i} \frac{v}{\abs{c}} \geq \left(\sum_{v \in \cT} c(v) \frac{v}{\abs{c}}\right) - \left(1 - \frac{1}{\alpha} \right) = 1 - \left(1 - \frac{1}{\alpha} \right) = \frac{1}{\alpha}.
\]
Since $\abs{c} \geq 1$ by our preprocessing, we conclude that player $i$ receives a total value of at least 
\[
 \sum_{j \in R_i} v_{ij} = \sum_{j \in J_i : j \in J_{c,v}^i} v \geq \sum_{j \in J_i : j \in J_{c,v}^i} \frac{v}{\abs{c}} \geq \frac{1}{\alpha}.
\]
A visualization of this argument is given in \Cref{fig:santa-to-makespan-approx}.

\end{proof}

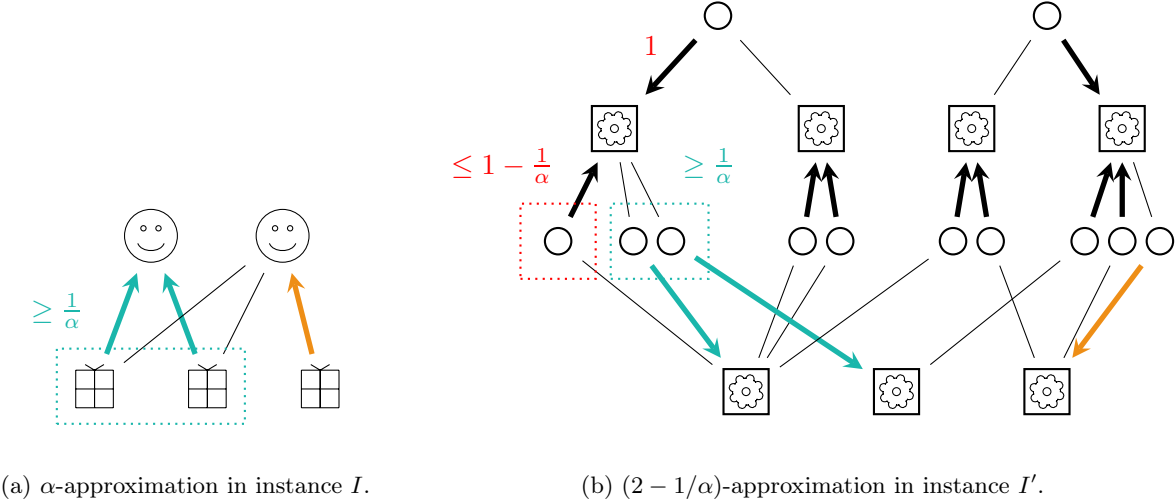
\begin{figure}
    \begin{subfigure}[b]{0.3\textwidth}
    \begin{center}
        \begin{tikzpicture}
            \node (c1) at (-1.25,5) {\child};
            \node (c2) at (0.5,5) {\child};
            \node (g1) at (-2,3) {\gift};
            \node (g2) at (-0.5,3) {\gift};
            \node (g3) at (1,3) {\gift};

            \draw[asg,col1] (g1) -- (c1);
            \draw[asg,col1] (g2) -- (c1);
            \draw[asg,col2] (g3) -- (c2);
            
            \draw[black]
                (c2) -- (g1)
                (c2) -- (g2);

            \draw [thick,col1, dotted] (-2.5,2.5) rectangle (0,3.5);
            \node[col1] at (-2.5,4) {$\geq \frac{1}{\alpha}$};
    \end{tikzpicture}
    \end{center}
    \caption{$\alpha$-approximation in instance $I$.}
    \end{subfigure}
    \begin{subfigure}[b]{0.7\textwidth}
    \begin{center}
    \begin{tikzpicture}
    \node (cj1) at (0.125,9) {\job};
    \node (m1) at (-1.25,7.5) {\machine};
    \node (m2) at (1.5,7.5) {\machine};
    \node (m1j2) at (-2,6) {\job};
    \node (m1j3) at (-1,6) {\job};
     \node (m1j4) at (-0.5,6) {\job};
    \node (m2j1) at (1.25,6) {\job};
    \node (m2j3) at (1.75,6) {\job};
    \draw[asg] (cj1) -- node[red, above left] {1} (m1);
    \draw[black] (cj1) -- (m2);
    \draw[asg] (m1j2) -- (m1);
    \draw[black] (m1) -- (m1j3);
    \draw[black] (m1) -- (m1j4);
    \draw[asg] (m2j1) -- (m2);
    \draw[asg] (m2j3) -- (m2);

    \node (cj2) at (4.5,9) {\job};
    \node (m3) at (3.5,7.5) {\machine};
    \node (m4) at (5.5,7.5) {\machine};
    \node (m3j1) at (3.25,6) {\job};
    \node (m3j3) at (3.75,6) {\job};
    \node (m4j1) at (5,6) {\job};
    \node (m4j2) at (5.5,6) {\job};
    \node (m4j3) at (6,6) {\job};
    \draw[black] (cj2) -- (m3);
    \draw[asg] (cj2) --  (m4);
    \draw[asg] (m3j1) --(m3);
    \draw[asg] (m3j3) -- (m3);
    \draw[asg] (m4j1) -- (m4);
    \draw[asg] (m4j2) -- (m4);
    \draw[black] (m4j3) -- (m4);

    \node (mr1) at (0.5,4) {\machine};
    \node (mr2) at (2.5,4) {\machine};
    \node (mr3) at (4.5,4) {\machine};

    \draw[black] (mr1) -- (m2j3);
    \draw[black] (mr1) -- (m2j1);
    \draw[black] (m1j2) -- (mr1);
    \draw[asg,col1] (m1j3) -- (mr1);
    \draw[asg,col1] (m1j4) -- (mr2);
    \draw[black] (mr1) -- (m3j1);
    \draw[black] (mr2) -- (m4j1);
    \draw[black] (mr3) -- (m4j2);
    \draw[asg,col2] (m4j3) -- (mr3);
    \draw[black] (mr3) -- (m3j3);

    \draw [thick,red, dotted] (-2.5,5.5) rectangle (-1.5,6.5);
    \draw [thick,col1, dotted] (-1.3,5.5) rectangle (0,6.5);
    \node[col1] at (0,7) {$\geq \frac{1}{\alpha}$};
    \node[red] at (-2.75,7) {$\leq 1-\frac{1}{\alpha}$};

    \end{tikzpicture}
    \end{center}
    \caption{$(2-1/\alpha)$-approximation in instance $I'$.}
    \end{subfigure}
    \caption{Visualization of the argument for translating approximate solutions used in \Cref{lemma:sc-to-makespan-approx}.}
    \label{fig:santa-to-makespan-approx}
\end{figure}

\subsection{Equivalence in the unrelated 2-value case}

In this section, we consider the two-value case of \santa and \makespan and prove that there exists an approximation preserving equivalence between these two problems.
We give a full proof in~\Cref{apx:two-value}. Here, we briefly restate the main theorem and highlight the technical ideas.

\twoValueEquivalence*

For the direction from two-value \santa to two-value \makespan, we show that we can apply a similar reduction as in~\Cref{thm:santa-to-makespan}. 
We do so by reducing the given \santa instance to an instance where the reduction does not introduce a third value. Further, we observe that, in the two-value case, we only have a polynomial number of relevant configurations, so we do not have to reduce this number and, thus, do not lose the additional factor of $1 +\epsilon$.

For the other direction, we prove that similar ideas as in the reduction of~\Cref{thm:santa-to-makespan} also work in the direction from \makespan to \santa if we only have to job sizes. This requires some additional ideas to ensure that we do not need to introduce further resource values.

\section{Reductions for matroid allocation problems}
\label{sec:restr-matroid}

We move to the restricted assignment setting and consider the matroid generalizations of \santa and \makespan. 
The main result here is \Cref{thm:restricted-matroid-equivalence}, which we restate here for convenience. 

\ThmReductionRestrictedMatroidSantaTwoValue*

We prove the theorem using the following two lemmas, one for each direction, and the same standard binary search framework as in the~\Cref{sec:santa-to-makespan}.
The key ideas for constructing instances and for transforming solutions in these reductions rely on polymatroid duality.
Note that, by~\Cref{prop:matroid-makespan-reduce-jobs}, we can w.l.o.g.~assume that the number of resources and jobs, respectively, is exactly two.

\begin{lemma}
Let $\alpha \geq 2$ and $I$ be an instance of the restricted \matroidmakespan problem with two jobs and $\opt(I) \leq 1$. Then we can compute an instance $I'$ of the restricted \matroidsanta problem with two resources, such that, given an $\alpha$-approximate solution for $I'$, we can compute a solution for $I$ with an objective value of at most $2 - 1 / \alpha$.
\end{lemma}

\begin{proof}
Given an instance $I$ with $\opt(I) \leq 1$ of the restricted \matroidmakespan problem with machines $E$ and two jobs with sizes $p_1,p_2 \geq 0$ and polymatroids $\PM_1,\PM_2$, we construct an instance~$I'$ of the restricted \matroidsanta problem as follows.

Let $k_1 = \lfloor 1 / p_1 \rfloor$ and let $k_2 = \lfloor 1 / p_2 \rfloor$. 
We first consider the polymatroids $\PM'_1 = \{x \in \PM_1 : x(e) \leq k_1 \; \forall e \in E \}$ and $\PM'_2 = \{x \in \PM_2 : x(e) \leq k_2 \; \forall e \in E \}$.
Let $f'_1$ and $f'_2$ be the associated submodular functions of these polymatroids. 
Since $\opt(I) \leq 1$, any optimal solution $x_j \in \cB(P_j)$ satisfies $x_j(e) \leq k_j$ for all $e \in E$, and therefore, $x_j \in \cB(P'_j)$, for $j \in \{1,2\}$. Thus, $f_j(E) = x_j(E) = f'_j(E)$.
Let $\dPM_j$ be the dual polymatroid of $\PM'_j$ with respect to the vector $k_j \cdot E$ (the vector of $\mathbb Z^E$ where all entries are equal to $k_j$), for $j \in \{1,2\}$.
We compose instance $I'$ using players $E$ and two resources with polymatroids $\dPM_1, \dPM_2$ and resource values $p_1, p_2$.

Let $t  = k_1 \cdot p_1 + k_2 \cdot p_2 - 1$. We show that $\opt(I') \geq t$. Fix an optimal solution for $I$ which selects bases $x_1 \in \cB(\PM_1)$ and $x_2 \in \cB(\PM_2)$. 
We define vectors $\overline{x}_j(e) = k_j - x_j(e)$ for all $e \in E$, and conclude 
that $\overline{x}_j \in\cB(\dPM_j)$, because $x_j \in \cB(\PM'_j)$. This means that $\overline{x}_1$ and $\overline{x}_2$ are a feasible solution for $I'$. Using $\opt(I) \leq 1$, for every player $e \in E$ holds
\[
    p_1 \cdot \overline{x}_1(e) + p_2 \cdot \overline{x}_2(e) = (1+t) -  \left(p_1  \cdot x_1(e) + p_2  \cdot x_2(e)\right) = (1+t) - \opt(I) \geq t,
\]
showing that $\opt(I') \geq t$.

We finally prove the stated bound on the objective value of an approximate solution. 
Fix an $\alpha$-approximate solution for $I'$ which selects bases $\overline{y}_1 \in \dPM_1$ and $\overline{y}_2 \in \dPM_2$.
We construct an approximate solution
$y_j \in \cB(\PM_j)$ for instance $I$ by setting $y_j(e) = k_j - \overline{y}_j(e)$ for every $e \in E$ and $j \in \{1,2\}$. The construction of the dual polymatroid $\dPM_j$ implies $y_j \in \cB(\PM'_j)$ for $j \in \{1,2\}$.
We further have $y_j \in \cB(\PM_j)$, because $\PM'_j \subseteq \PM_j$ and $y_j(E) = f'_j(E) = f_j(E)$.
Moreover, for every machine $e \in E$ holds
\[
p_1 \cdot y_1(e) + p_2 \cdot y_2(e) = (1+t) - \left(p_1 \cdot \overline{y}_1(e) + p_2 \cdot \overline{y}_2(e) \right) \leq (1+t) - \frac{1}{\alpha} \cdot  \opt(I') \leq 1+t - \frac{t}{\alpha}.
\]

Since by construction $t \le 1$, we have $t - t / \alpha \le 1 - 
 1 / \alpha$, which implies that the makespan of the constructed solution $(y_1,y_2)$ is at most $2 - 1 / \alpha$.
\end{proof}

The second direction can be shown with the same proof idea and some additional tweaks specific to the direction. We give the full proof for the following lemma in~\Cref{apx:restr-matroid}

\begin{restatable}{lemma}{lemMatroidReductionSM}
Let $\alpha \geq 2$ and $I$ an instance of the restricted \matroidsanta problem with two resources and $\opt(I) \geq 1$. Then we can compute an instance $I'$ of the restricted \matroidmakespan problem with two jobs, such that, given a $(2- 1 / \alpha)$-approximate solution for~$I'$, we can compute a solution for $I$ with an objective value of at least $1 / \alpha$.
\end{restatable}

\section{Local search algorithm}
In this section we present our algorithm for Theorem~\ref{th:main-approx} that finds
an $(8+\epsilon)$-approximation for the restricted \matroidsanta problem and a $(4 + \epsilon)$-approximation in the case of two values. The analysis is moved to Appendix~\ref{apx:ana-localsearch}

We will show in Lemma~\ref{lemma:two-value-santa-to-general-matroid} that it suffices to solve the following problem with $\alpha = 4 + \epsilon$.
Given a matroid $\cM = (E, \cI)$ and a polymatroid $\PM$ over the same set of elements as well as some $b\in \mathbb N$, find some $I_M\in \cI$ and $y \in \PM$
such that for every $i\in E$ we have $i\in I_M$ or $y(i) \ge b$, or determine that no solution exists for $\alpha b$.
Before we move to the algorithm, we define some specific notation used throughout the section.
\paragraph*{Notation.}
Recall that for some $X\subseteq E$, we write $b \cdot X$ as the vector $y\in \mathbb Z^E$ with
$y(i) = b$ for $i\in X$ and $y(i) = 0$ for $i\notin X$.
Contracting a set $X \subseteq E$ of a matroid $\cM$ defines a new matroid $\cM \slash X$ obtained from
restricting the elements to $E\setminus X$ and defining
the rank function as $r(Y \mid X) = r(Y \cup X) - r(X)$.
Naturally, the independent sets of the contracted matroid form all the sets that together with \emph{any} independent set in $S$ are independent in the original matroid.
We use contraction primarily to fix some elements in
the matroid.
We need a similar notion for the polymatroid $\PM$ defined
by the submodular function~$f$.
However, while a single element $i\in E$
has rank function $r(i) \le 1$ in a matroid, the value $f(i)$ could be arbitrarily large.
In particular, contracting a set using $f(Y \mid X)$
we may reserve more resources for $X$ than 
intuitively necessary (recall we only need to cover elements with the polymatroid $b$ times).
Hence, we need to use a more sophisticated approach here. First, consider
the polymatroid $\PM' = \{y\in \PM : y(i) \le b \ \forall i\in X\}$, i.e., a
restriction on the multiplicity of each element in $X$ within the polymatroid. 
Let, $f'$ be the submodular function defining $\PM'$.
After this transformation we can use $f'(Y \mid X)$ without the
aforementioned issues. We introduce the short notation
\begin{equation*}
    f(Y \mid b \cdot X) = f'(Y \mid X) \ .
\end{equation*}
Note that $f(Y \mid b \cdot X)$ behaves as one would expect in the
sense of decreasing marginal returns, see Lemma~\ref{lem:decr-marg}.
We will define both $r(Y \mid X)$ and $f(Y \mid b \cdot X)$ on
all $X\subseteq E$ instead of only $Y\subseteq E\setminus X$.
More precisely, $r(Y \mid X) = r(Y \cup X) - r(X) = r((Y\setminus X) \cup X) - r(X) = r(Y \setminus X \mid X)$ gives a natural extension
although clearly the elements in $X$ behave
trivially.
Similar to this, we extend $f(Y \mid b \cdot X)$ to $Y \cap X \neq \emptyset$ by defining $f(Y \mid b \cdot X) = f(Y\setminus X \mid b \cdot X)$.

\paragraph*{Framework.}
We move to a further variation of the problem, which resembles an augmentation framework
similar to matroid partition problems, where we are given a partial solution that we then extend.
The algorithm itself is defined using recursion with the following interface.
\begin{description}
\item{Input.} Matroid $\cM = (E, \cI)$ with rank function $r$, polymatroid $\PM \subseteq \mathbb Z^E$ with function $f$, both over the same elements, and a number $b\in\mathbb N$.

Further, disjoint sets $I_M\in \cI, b\cdot I_P \in \PM, B_0 \subseteq E$.
Finally, a partial order~$\prec$ on $B_0$.

\item{Output.} Either an augmented solution
$I_M'\in \cI$ and $b\cdot I'_P\in \PM$
such that $I_M'\dot\cup I_P' \supseteq I_M\cup I_P$ and
$|I_M' \cap B_0| \ge \epsilon^2 |B_0|$
or ``failure''.

In case failure is returned, we provide a certificate that proves that no $I_M^*, I_P^*$ can exist with $I_M^*\cup I_P^* \supseteq I_M\cup I_P$, $I_M^*\in \cI$ and the stronger conditions $\alpha b \cdot I_P^* \in \PM$ for $\alpha = 4 + \cO(\epsilon)$ and $|I_M^* \cap B_0| \ge 3 \epsilon |B_0|$. Details on the certificate follow in the analysis.

The partial order~$\prec$ affects which elements of $B_0$ the algorithm tries first to add to $I_M$. The precise guarantees on the output are subtle, but important inside the recursion, see proof of Lemma~\ref{lem:addable3}.
\end{description}
We can apply this variant to solve our previous polymatroid problem as follows: we initialize $I_P$ as the set of all elements $i\in E$ that have $r(i) = 0$.
If $b\cdot I_P \notin \PM$ it is clear that the optimum is smaller than $b$.
Assuming $b\cdot I_P \in \PM$ we set $\cI = \emptyset$ and now extend $I_M\cup I_P$ one element at a time by calling the procedure above with $B_0 = \{i\}$ for some element $i\notin I_M\cup I_P$. Each time $I_M$ and $I_P$ will be changed, but end up covering
an additional element. Repeating this at most $|E|$ times we either have a solution
that covers all elements or
some element $i$ cannot be added, which certifies that no solution exists with $i$ covered by the matroid. In this
case we alter the rank function to $r(X) \gets r(X - i)$, setting in particular $r(i) \gets 0$. We then restart the whole procedure.

We will now describe how to solve this variant of the problem.
First, we assume without loss of generality that $I_M \cup I_P \cup B_0 = E$ by
simply dropping irrelevant elements from the input.
As stated above, we want to add many elements of $B_0$ to $I_M$.
In the trivial case that $r(B_0 \mid I_M) \ge \epsilon^2 |B_0|$, we add greedily
as many elements of $B_0$ as possible to $I_M$ (while maintaining $I_M\in \cI$), which will
result in $|I_M\cap B_0| \ge \epsilon^2 |B_0|$, and we
terminate successfully.
Otherwise, we will have to remove elements from $I_M$ before we can add sufficiently many elements of $B_0$ to $I_M$.
To this end, we will carefully construct a set
of \emph{addable elements} $A\subseteq I_M$, where the notion has historical reasons and comes from the idea
that we want to ``add'' $A$ to $I_P$.
The procedure for creating $A$ is deferred to later
and here we summarize only its important properties.
The existence of $A$ will be guaranteed by the algorithm.
Along with $A$ we also create the set $C$ with $A \subseteq C \subseteq I_M$, which contains more elements
of $I_M$ that are relevant for adding
$B_0$ to $I_M$ (but not all of them could be added to $A$).
The relevant properties of $A$ and $C$ are as follows.
\begin{enumerate}
    \item It holds that $2b \cdot A \in \PM$, which means that in principle $A$ could be added to $I_P$ (and removed from~$I_M$).
    Note that this does not take into account potential conflicts with other elements currently in~$I_P$.
    Also remark that we are intentionally overprovisioning here, by using $2b$ instead of~$b$.
    \item Set $A$ is maximal within $C$ regarding the previous property.
    More specifically,
    $f(i \mid 2b \cdot A) < 2b$ for all $i\in C\setminus A$.
    \item If we remove many elements of $A$ from $I_M$, we are able to add many elements of $B_0$ to $I_M$.
    Specifically, we require that for every $R\subseteq A$ with $|R| \ge \epsilon |A|$
    we have 
    \begin{equation*}    
    r(B_0 \mid I_M \setminus R) \ge \epsilon^2 |B_0| \ .
    \end{equation*}
    We note that while this property initially holds, only a weaker version is maintained as the algorithm progresses. For more details see Lemma~\ref{lem:addable3}.
    \item Set $C$ should contain almost all elements that
    block elements of $B_0$ from being added to $I_M$.
    Formally, $r(B_0 \mid C) \le 2\epsilon |B_0|$. In particular, a solution that covers
    many elements of $B_0$ with the matroid 
    needs to cover a substantial amount of elements
    in $C$ with the polymatroid.
\end{enumerate}
Our new goal is to move many elements of $A$ to $I_P$, which may not be possible immediately because of conflicting
elements currently in $I_P$. We first characterize these elements:
Define the \emph{blocking elements} $B$ as the set of all $i\in I_P$ such that $f(i \mid b \cdot (I_P \cup A - i)) < b$.
It is intuitively clear that elements not in $B$ are not relevant to adding $A$:
assume we remove some elements of $I_P$ and add some elements of $A$ to it. Then afterwards all elements not in $B$ can easily be added back to $I_P$ (if they were removed), since each of their marginal values will still be at least $b$. 

When an element in $A$ can be added to $I_P$, we will not add it right away.
Instead we will
only place it in a set of \emph{immediately addable elements} $A_I$. Only when
we have enough of these elements to successfully terminate, we will add them to $I_P$.
This is mainly for simplicity, i.e., to keep structures as static as possible during execution.
We now repeatedly perform the first possible operation from the following:
\begin{enumerate}
    \item If $f(i \mid b \cdot (I_P \cup A_I)) \ge b$ (equivalently, $b \cdot (I_P \cup A_I + i) \in \PM$) for some $i\in A\setminus A_I$, add $i$ to $A_I$.
    \item If $|A_I| \ge \epsilon |A|$, add $A_I$ to $I_P$, remove it from $I_M$, and greedily add as many elements of $B_0$ as possible to $I_M$,
    in the order given by~$\prec$. We can then terminate successfully (Lemma~\ref{lem:addable3}).
    \item If $|B| < \epsilon |B_0|$, return ``failure''.
    \item If none of the above applies, we will recurse on $B$, which means that we try to move many elements of $B$ to $I_M$ so that they can be removed from $I_P$, hopefully allowing us to move elements of $A$ to $A_I$. The details of the recursion follow towards the end of the section.
\end{enumerate}
\paragraph*{Construction of addable elements.}
We construct a series of disjoint sets $A_1,A_2,\dotsc$
as follows:
assume that $A_1,A_2,\dotsc,A_{\ell-1}$ have already been
created. We initialize $A_\ell = \emptyset$.
Then repeat the following until exhaustion:
if there exists an $i\in I_M\setminus (A_1\cup A_2 \cup \cdots \cup A_\ell)$ with $r(i \mid B_0 \cup I_M \setminus A_{\ell} - i) = 0$ and
\begin{equation*}
    f(i \mid 2b \cdot (A_1 \cup A_2 \cup \cdots \cup A_\ell)) \ge 2b \ ,
\end{equation*}
then we add $i$ to $A_{\ell}$.
If $|A_{\ell}| < \epsilon |B_0|$ we terminate with
$A = A_1 \cup A_2 \cup \cdots \cup A_{\ell-1}$. 
Otherwise, we continue with the next set.
When the construction of $A$ is finalized,
we define
\begin{equation*}
  C = B_0 \cup A \cup \{i\in I_M\setminus A : f(i \mid 2b \cdot A) < 2b \} \ .
\end{equation*}
\paragraph*{Recursion.}
We will denote the input of the recursion using prime next to the symbol, e.g. $E'$, $B'_0$, etc.
The goal of the recursion is to move many elements of $B$ to $I_M$.
On the other hand, we want to keep our structures
within $B_0, A,$ and $C$ largely intact.
To this end, we first contract $C$ from the matroid,
which means that the recursion cannot remove elements
of $C$ from $I_M$. It may not be intuitively clear
why this would be bad for us, but removing many elements of $C$ does not necessarily allows us to add many elements of $B_0$ (comparable to Property~3 of the set $A$) for arbitrary elements of $C$.
In any case, we want to avoid
the complications related to such changes in $C$.
Hence, we set
\begin{align*}
    E' &= E\setminus C\ , \\
    I_M' &= I_M \setminus C \ ,\text{ and} \\
    r'(X) &= r(X \mid C) \quad \forall X\subseteq E'
\end{align*}
which defines a new matroid $\cM' = (E', \cI')$.
Regarding the polymatroid, we remove $B$ from $I_P$ so as to produce an input where
sets $B'_0$, $I_M'$, and $I_P'$ are disjoint. However, by contracting $b \cdot B$ from the polymatroid we make sure that we can add it back to the modified solution after the recursion has returned (at least for those elements that were not moved into $I_M$).
Furthermore, we want to avoid that the recursion moves elements to $I_P$
that hinder $A$ from being added to the polymatroid. Hence, we contract $b\cdot A$
as well.
This is achieved by setting
\begin{align*}
    I'_P &= I_P \setminus B \text{ and} \\
    f'(X) &= f(X \mid b \cdot (A \cup B)) \quad \forall X\subseteq E' \ .
\end{align*}
From $f'$ we obtain the new polymatroid $\PM'$. Finally, we define
\begin{equation*}
   B'_0 = B_0 \cup B
\end{equation*}
and extend~$\prec$ by giving all elements of $B$ lower priority than $B_0$.
Intuitively, it does not hurt us to include $B_0$ in $B'_0$.
If the recursion manages to move elements to $B_0$,
this is only good for us.
We prove in Lemma~\ref{lem:recursive-feasible} that
this indeed forms a valid input of the problem.
If the recursive call returns failure, we return failure as well.
Otherwise, we update as follows. Let $I''_M$ and $I''_P$ be the output of the recursion.
First, we add back the previously removed $C$:
\begin{equation*}
    I_M \gets C \cup I_M''  \ .
\end{equation*}
As for $I_P$, we want to add back $B$ except for those elements that were covered with the matroid in the recursive call. Thus,
\begin{equation*}
    I_P \gets (B \setminus I_M'') \cup I_P'' \ .
\end{equation*}
In Lemma~\ref{lem:inv-feasible} we show that the new $I_M, I_P$ constitute
again a feasible solution.
If the returned set~$I_M''$ satisfies $|B_0 \cap I_M''| \ge \epsilon^2 |B_0|$ we terminate successfully.
Otherwise, it must hold that $|B\cap I_M''| \ge \epsilon^2 |B|$, which intuitively means we made big progress in freeing $B$ and is used in the running time analysis.
Since $I_P$ has changed, $B$ may no longer correspond to its initial definition. Hence,
we update $B$ to again reflect
the set of all $i\in I_P$
with $f(i\mid b \cdot (I_P\cup A - i)) < b$ according to the
new set $I_P$. 

\section{Final remarks} 

For the two notorious open problems in scheduling theory, we prove \makespan to be at least as difficult as \santa; more precisely, a better-than-$2$ approximation for \makespan would imply an $\cO(1)$-approximation for \santa. In the two-value case both problems appear equivalent w.r.t.\ approximability. 
The obvious open question is whether there is also a \makespan-to-\santa reduction (for restricted assignment or the general case). Here we note that for restricted assignment \makespan,
all efforts to refine the local search method in order
to give a better-than-$2$ approximation have failed so far. 
Also with our new reduction techniques it seems that it would require additional ideas to handle this problem.
By the reductions, our local search method generalizes all previously known polynomial-time local search results for the two problems (up to constants) and
yields the first approximation algorithm for a new matroid \santa variant, where items are allocated to the basis of a (poly-)matroid.
We hope that this makes it clearer where the power and limitations of the method are.

Finally, we comment on an alternative matroid scheduling variant with matroid constraints on the {\em items} allocated to a specific machine/player. In the {\em machine-matroid} \makespan problem, each machine would be given a matroid on the jobs. 
All jobs must be assigned such that each machine receives an independent set of its matroid. 
The player-matroid \santa can be defined similarly. 

Kawase et al.~\cite{KawaseKMS21} consider such matroid partition problems for various objective functions showing complexity results. 
Further, two special-matroid examples for \makespan have been studied, namely, bag-constrained scheduling~\cite{DasW17,GrageJK19} (single partition matroid) and scheduling with capacity constraints~\cite{ChenJLZ16} (uniform matroids). 
The approximability lower bound $\Omega((\log n)^{1/4})$ 
by \cite{DasW17} holds for the restricted assignment setting and even translates to an inapproximability bound for machine-matroid \makespan for identical machines with machine-dependent matroids. 
We are not aware of any similarly strong lower bounds for the \santa variant.

\bibliographystyle{plain}
\bibliography{literature}

\appendix

\section{Definitions and notation}
\label{apx:notation}

\propreducejobs*

\begin{proof}
In the following, we use the notation $[h] := \{1,\ldots,h\}$.
Let $I$ be an instance of the restricted \matroidmakespan problem with machines $E$ and $h$ distinct processing times $p_1,\ldots,p_h$.
Let $J_\ell$, $\ell \in [h]$, denote the set of jobs with processing times $p_\ell$.
Further, let $\PM^\ell_j$ with $\ell \in [h]$ and $j \in J_\ell$ denote the corresponding polymatroids over $E$ and let $f_j^\ell$ be the associated submodular function.

We construct an instance $I'$ of the restricted \matroidmakespan problem with $h$ jobs by using the same set of machines $E$ and creating the polymatroids $\PM_\ell$ for $\ell \in [h]$ with the montone submodular function $f_\ell(S) = \sum_{j \in J_\ell} f_j^\ell(S)$ for every subset $S \subseteq E$.
Note that $\PM_\ell = \sum_{j \in J_\ell} \PM^\ell_j$~\cite{schrijver2003combinatorial}.
For $\ell \in [h]$, the goal in instance $I'$ is to find vectors $x_\ell \in \mathcal{B}(\PM_\ell)$ such that $\max_{e \in E} \sum_{\ell \in [h]} p_\ell \cdot x_\ell(e)$ is minimized. We prove that this reduction preserves the approximation factor.

Consider a solution of instance $I$ that selects the bases $x_j^\ell$ for job $j \in J_\ell$ with $\ell \in [h]$ and consider the vectors $x'_\ell$ with $x'_\ell(e) = \sum_{j \in J_\ell} x_j^\ell(e)$ for all $e \in E$.
Using again the fact that $\PM_\ell = \sum_{j \in J_\ell} \PM^\ell_j$ we have $x'_\ell \in \PM_\ell$ for all $\ell \in [h]$.
In particular, $x'_\ell(E) = \sum_{e \in E} \sum_{j \in J_\ell} x_j^\ell(e) = \sum_{j \in J_\ell} f_j^\ell(E) = f_\ell(E)$, so $x'_\ell$ is a basis of $\PM_\ell$. Thus, $(x'_1, \ldots,x'_h)$ is a feasible solution for instance $I'$.
Furthermore,
\[
\opt(I') \le \max_{e \in E}\sum_{\ell \in  [h]} x'_\ell(e) \cdot p_\ell = \max_{e \in E} \sum_{\ell \in [h]} \sum_{j \in J_\ell} x_j^\ell(e) \cdot p_\ell = \opt(I).
\]

Consider some solution $(y'_1, \ldots , y'_h)$ to instance $I'$, i.e., $y'_\ell \in \cB(\PM_\ell)$ and  $f_\ell(E) = y'_\ell(E)$ for all $\ell \in [h]$.
We construct a solution to $I$ by decomposing each $y'_\ell$, $\ell \in [h]$, into bases $y_j^\ell \in \cB(\PM_j^\ell)$ such that $y'_\ell(e) = \sum_{j \in J_\ell} y_j^\ell(e)$ holds for all $e \in E$.
As $\opt(I') \le \opt(I)$, this implies that the reduction preserves the approximation factor. 
If such decomposition would not exist for some $\ell \in [h]$, then, by construction of the submodular function $f_\ell$, we would arrive at a contradiction to $f_\ell(E) = y'_\ell(E)$.

To find the decomposition for an $\ell \in [h]$ in polynomial time, consider the polymatroids $\bhat{\PM}^\ell_j$ which are just copies of the original polymatroids $\PM^\ell_j$ on pairwise disjoint copies $\bhat{E}_j$ of the ground set $E$. 
For each $\ell \in [h]$, we decompose the solution $y'_\ell$ of instance $I'$ into bases of the copy polymatroids, which then implies a decomposition into bases of the original polymatroids. 
For each $e \in E$, let $C_e$ denote the set of copies of $e$ introduced by the ground set copies.
We want to find a basis $\bhat{y}_j^\ell$ for every $j \in J_\ell$ such that $\sum_{\bhat{e} \in C_e} \bhat{y}_j^\ell(\bhat{e}) = y'_\ell(e)$ holds for all $e \in E$ and $\ell \in [h]$. 
For an element $e \in E$ and $\ell \in [h]$, consider the polymatroid $\mathcal{X}^\ell_e$ on the ground set $C_e$ implied by bases $\mathcal{B}(\mathcal{X}^\ell_e) = \{x \in \mathbb{Z}^{C_e}_{\geq 0} : x(C_e) = y'_\ell(e)\}$ and let $\mathcal{X}^\ell$ denote the union of these polymatroids. Furthermore, let $\bhat{\PM}^\ell$ denote the union of the polymatroids $\bhat{\PM}^\ell_j$. The largest element in the intersection of $\mathcal{X}^\ell$ and $\bhat{\PM}^\ell$ gives us the decomposition. We can compute the largest element in the intersection in polynomial time using algorithms for polymatroid intersection (cf.~e.g.~\cite[Chapter 41]{schrijver2003combinatorial}).

The statement for \matroidsanta can be shown with the same reduction and proof; only the inequality $\opt(I') \le \opt(I)$ trivially changes to $\opt(I') \ge \opt(I)$.
\end{proof}

\section{Santa Claus with polynomially many configurations}\label{apx:polynomial-configs}

\santaConfigurationReduction*

\begin{proof}
    Let $\epsilon > 0$ be a sufficiently small constant and $\kappa = \ceil{1/ \epsilon^3}$.
    Given a \santa instance $I$ with the set $P$ of $m$ players, the set $R$ of $n$ resources and $\opt(I) \geq 1$, we construct the \santa instance $I'$ by executing the following steps. 
    
    \begin{enumerate}
        \item Use the same set of players and resources as in $I$.
        \item Round all resource values $v_{ij}$ down to the closest power of $1 / (1+\epsilon)$. That is, we round $1/(1+\epsilon)^{\ell-1} \ge v_{ij} \ge 1 / (1+\epsilon)^\ell$ to $\bar{v}_{ij} = 1 / (1+\epsilon)^\ell$.  If $v_{ij} \ge 1$, then we set $\bar{v}_{ij} = 1$. Furthermore, we round all $v_{ij}$ with $v_{ij} < 1 / ((1+\epsilon)n)$ to zero.
        In summary, each  $\bar{v} \in \cT$ is either a power of $1/(1+\epsilon)$ of value at least $1/((1+\epsilon)n)$ or $0$.
    \end{enumerate}

    Next, we construct the configurations for the rounded instance $I'$ by executing the following steps. 
    Since these steps will reduce the number of possible configurations per player to a polynomial, an algorithm creating the configurations can just compute them via enumeration.  
    
    \begin{enumerate}
         \setcounter{enumi}{2}
         \item For each player $i \in P$, we create the set $\cC_i$ of configurations and restrict the set to configurations $c$ such that, for every value type $\bar{v} \in \cT$, either $c(\bar{v}) = 0$, $c(\bar{v}) = \lceil(1+\epsilon)^\ell\rceil$ or $c(\bar{v}) = \lfloor(1+\epsilon)^\ell\rfloor$ for some $\ell \in \mathbb{N}_0$ with $(1+\epsilon)^\ell \le n$. 
        \item 
        Let $\bv_1 \geq \ldots \geq \bv_\tau$ be the rounded value types in $\cT$. We partition $\cT$ into $\kappa$ value classes $\cT_1,\ldots,\cT_\kappa$ where $\cT_\ell = \{\bv_{\ell + s \cdot \kappa} : s=0,1,\ldots\}$.
        We further restrict the set of configurations $\cC_i$ for a player $i \in P$ to configurations $c$ which satisfy
        for every $1 \leq \ell \leq \kappa$ and
        for every $\bv, \bv' \in \cT_\ell$ with $\bar{v} > \bar{v}'$ that either $c(\bar{v}) < c(\bar{v}')$ or $c(\bar{v}) = 0$ or $c(\bar{v}') = 0$. That is, the function values of value types $\bar{v} \in \cT_\ell$ of one value class that actually occur in a configuration (i.e., have $c(\bar{v}) > 0$) increase with decreasing value $\bar{v}\in \cT_\ell$.
    \end{enumerate}
  
    We first argue that, for each player $i$, the number of configurations in $\cC_i$ is polynomial in the input size. Because of the second step, the number of value types in $I'$ is in $\cO_\epsilon(\log n)$. By the third step, the number of distinct function values $c(\bar{v})$ over all configurations $c \in \cC_i$ and all  value types $\bar{v} \in \cT$ is in $\cO_\epsilon(\log n)$ as well.
    
    By step $4$, we can represent the entries of $c \in \cC_i$ which correspond to the same value class $\cT_\ell$ in terms of a vector with $\cO_\epsilon(\log n)$ entries such that all non-zero entries strictly increase in value (each entry of the vector corresponds to a value type $\bar{v} \in \cT_\ell$, in decreasing order, and the entry values represent the corresponding function values $c(\bar{v})$).
    We can represent such a vector by the set of entries that have a non-zero value and by the set of non-zero values that occur in the vector.
    Since there are at most $2^{\cO_\epsilon(\log n)}$ different sets of non-zero values that can occur in the vector and at most $2^{\cO_\epsilon(\log n)}$ different sets of non-zero entries, the number of such vectors is $2^{\cO_\epsilon(\log n)} \cdot 2^{\cO_\epsilon(\log n)} \subseteq n^{\cO_\epsilon(1)}$. 
    Since the total number of value classes is constant, and there is a simple rule to compose the vectors of every class to a vector for all value types $\cT$, the total number of vectors which represent every valid configuration is polynomial in the size of $I$.
    
    Next, we prove the approximation factor. Let $I'$ denote the rounded instance constructed by the first two steps.
    Furthermore, let $\cC'$ denote the collection of configurations that is created by only executing the third step of the construction and let $\cC$ denote the final collection of configurations.
    We separately prove $\opt_{\cC'}(I') \ge 1 / (1+\epsilon)^3$ and $\optc(I') \ge  \opt_{\cC'}(I') / (1+\epsilon)$. 
    Together, these inequalities imply $\optc(I') \ge 1 / (1+\epsilon)^4$. Then, for any sufficiently small $\epsilon' > 0$, we can choose $\epsilon = \epsilon'/5$ and conclude $\optc(I') \ge 1 / (1+\epsilon')$.

    We first show $\opt_{\cC'}(I') \ge 1 / (1+\epsilon)^3$. Consider some optimal solution for $I$. For a player $i$, let $A_i$ denote the resources that are assigned to $i$ in the optimal solution. Clearly $v(A_i) \ge 1$. Discarding all resources in $A_i$ with value smaller than $1/((1+\epsilon)n)$ reduces the value of $A_i$ by a factor of at most $1+\epsilon$. Rounding the remaining resource values down to powers of $1+\epsilon$ reduces the value by another factor of $1+\epsilon$.  To make sure that the remaining resources in $A_i$ with their rounded values match a configuration in $\cC'_i$, we might have to remove a $1+\epsilon$ fraction of the resources for each value type from $A_i$. This reduces the value of $A_i$ by another factor of $1+\epsilon$. 
    The remaining value is at least $1 / (1+\epsilon)^3$. By doing this for every player $i$, we obtain a solution to $I'$ that matches $\cC'$ with an objective value of at least $1 / (1+\epsilon)^3$, which implies $\opt_{\cC'}(I') \ge 1 / (1+\epsilon)^3$. 

    Finally, we prove $\optc(I') \ge  \opt_{\cC'}(I') / (1+\epsilon)$. To that end, fix an optimal solution for $I'$ among those solutions that match $\cC'$. Consider some player $i$ and let $c'_i$ denote the configuration that is selected for player $i$ in the optimal solution for $I'$. We argue that we can find a configuration $c_i \in \cC_i$ that 
    \begin{enumerate}[(i)]
        \item has a total value that is at least a $1 / (1+\epsilon)$ fraction of the value of $c'_i$ and 
        \item satisfies $c_i(\bar{v}) \le c'_i(\bar{v})$ for all $\bar{v} \in \cT$.
    \end{enumerate}
    This gives us a feasible solution for $I'$ that matches $\cC$ and has an objective value of at least $\opt_{\cC'}(I') / (1+\epsilon)$ and, thus, proves the statement. 

    We start by building a configuration $c_i$ independently for every value class $\cT_\ell$.
    
    First, we iteratively construct a subset $S_\ell \subseteq \cT_\ell$ of value types as follows.
    
    Start with the largest $\bv \in \cT_\ell$ such that $c'_i(\bv) > 0$ and add $\bv$ to $S_\ell$. Then, find the largest $\bv' \in \cT_\ell$ with $\bv > \bv'$ and $c'_i(\bv) < c'_i(\bv')$. Add $\bv'$ to $S_\ell$ and repeat from $\bv'$ until we do not find another value type to add. Based on $S_\ell$, define configuration $c_i$ as $c_i(\bar{v}) = c'_i(\bar{v})$ if $\bv \in S_\ell$ and  $c_i(\bar{v}) = 0$ otherwise. 
    By choice of the sets $S_\ell$, the configuration $c_i$ is contained in $\cC_i$ and satisfies (ii).
    
    It remains to prove that $c_i$ also satisfies (i).  
    Fix an arbitrary value type $\bv_j \in S_\ell$ and let $\bv_{j'}$ denote the next smaller value type in $S_\ell$. 
    Recall that the rounded value types are indexed in decreasing order and let $s$ be the integer such that $j' = j + \kappa \cdot (s + 1)$. If $\bv_j$ is already the smallest value type in $S_\ell$, we set $s$ to the largest integer such that $\ell + \kappa \cdot s \leq \tau$. 
    
    We show  
    \begin{equation}
        c_i(\bar{v}_j) \cdot \bar{v}_j \ge \frac{1}{1+\epsilon} \cdot \sum_{s' = 0}^{s} c'_i(\bar{v}_{j + s' \cdot \kappa}) \cdot \bar{v}_{j + s' \cdot \kappa}. \label{eq:inc-config-rounding}
    \end{equation}
    If this inequality holds for all  $\bv_j \in S_\ell$, and for all $1 \leq \ell \leq \tau$, then (ii) follows. 

    To prove the inequality, observe that, by choice of set $S_\ell$, all integers  $0 \leq s' \leq s$ satisfy $c'_i(\bar{v}_{j + s' \cdot \kappa}) \le c'_i(\bar{v}_{j})$. Furthermore, $\bar{v}_{j + s' \cdot \kappa} = \bar{v}_j / (1+\epsilon)^{\kappa \cdot s'}$ by the rounding of the value types. This gives us
    \begin{align}
        \sum_{s' = 0}^{s} c'_i(\bar{v}_{j + s' \cdot \kappa}) \cdot \bar{v}_{j + s' \cdot \kappa} 
        \leq c'_i(\bar{v}_{j}) \cdot \sum_{s' = 0}^{s} \bar{v}_{j + s' \cdot \kappa} 
        &= c_i(\bar{v}_{j}) \cdot  \bar{v}_{j} \cdot \sum_{s' = 0}^{s} \frac{1}{(1+\epsilon)^{\kappa \cdot s'}} \notag \\
        &\leq c_i(\bar{v}_{j}) \cdot \bar{v}_{j} \cdot \frac{1}{1 - \frac{1}{(1+\epsilon)^\kappa}} \label{eq:geometric-bound}
    \end{align}

    By further using 
    \[
        \kappa \geq \frac{1 + \epsilon}{\epsilon^2} \geq \frac{1}{\epsilon} \cdot \frac{1}{\ln(1+\epsilon)} \geq \frac{\ln(1 + \frac{1}{\epsilon})}{\ln(1+\epsilon)} = \log_{1+\epsilon} \left(\frac{1+\epsilon}{\epsilon} \right),
    \]
    it is not hard to see that \eqref{eq:geometric-bound} is at most $c_i(\bar{v}_{j}) \cdot  \bar{v}_{j} \cdot (1 + \epsilon)$, and thus, implies \eqref{eq:inc-config-rounding}. This concludes the proof of $\optc(I') \ge \opt_{\cC'}(I') / (1+\epsilon)$.
\end{proof}

\section{Approximation equivalence for the unrelated two-value case}
\label{apx:two-value}

We consider the two-value  \santa and \makespan variants, where all resource values and job sizes are in $\{\si,\bi,0\}$ and $\{\si,\bi, \infty \}$, respectively, with $\si,\bi \geq 0$. Assume w.l.o.g.~that $\si \le \bi$. We prove the following theorem.

\twoValueEquivalence*

To prove this theorem we prove two lemmas (\Cref{thm:unrelated-twovalue-makespan-to-santa} and \Cref{theorem:unrelated-twovalue-santa-to-makespan}) in the following subsections, one for each direction.
Then, a standard binary search argument completes the proof.

\subsection{Makespan to Santa Claus}

\begin{lemma}\label{thm:unrelated-twovalue-makespan-to-santa}
    Let $I$ be an instance of the two-value \makespan problem with $\opt(I) \leq 1$. For any $\alpha \geq 2$, 
    we can construct in polynomial time an instance $I'$ of two-value \santa such that, given an $\alpha$-approximate solution for $I'$, we can compute in polynomial time a solution for $I$ with an objective value of at most $2-1/\alpha$.
\end{lemma}

Fix an instance $I$ of the two-value \makespan problem with $\opt(I) \leq 1$. 
First, we observe that we can w.l.o.g. assume that $\bi > \opt(I) / 2$. Otherwise, the algorithm by~\cite{lenstra1990approximation} gives us a solution of objective value at most $\opt(I) + \bi \le 3/2 \cdot \opt(I)$. 
Since $\alpha \ge 2$, this solution satisfies the lemma for every possible $\alpha$. Thus, in the following we assume $\bi > \opt(I) / 2$. 

\paragraph{Construction.}
We construct an instance $I'$ of the two-value \santa problem.
Let $k = \min\{\lfloor 1 / \si \rfloor, n\}$, where $n$ is the number of jobs in $I$. That is, $k$ denotes the maximal number of small jobs that can be assigned to a single machine in an optimal solution with $\opt(I) \leq 1$. By our assumption on the size $\bi$,  we have that at most one big job can be placed on a single machine. 

For every machine $i$ we introduce a machine-player~$q_i$, one (large) resource~$r_i$, and~$k$ (small) resources~$r_i^1,\ldots,r_i^{k}$. 
The value of the large resource~$r_i$ for player~$q_i$ is equal to~$\bi$, the value of a small resource~$r^{\ell}_i$ for player~$q_i$ is equal to~$\si$.

For every job $j$, we introduce a job-player $\bhat{q}_j$. Furthermore, for every machine $i$, we set the value of resource $r_i$ for $\bhat{q}_j$ to $\bi$ if $p_{ij} = \bi$ and to $0$ otherwise.
For a small resource~$r_i^\ell$ we set the value for $\bhat{q}_j$ to $\bi$ if $p_{ij} = \si$ and to $0$ otherwise.

Note that in  $I'$, every machine-player $q_i$ has only values in $\{0, \si, \bi \}$, and every job-player $\bhat{q}_j$ has only values in $\{0, \bi\}$. Thus, $I'$ is a two-value \santa instance. Further, for every machine the number of introduced resources is at most the total number of jobs in $I$, asserting that $I'$ is of polynomial size.
An illustration of this construction is depicted in \Cref{fig:2-makespan-to-2-santa-construction}.

\begin{figure}
    \begin{subfigure}[b]{0.4\textwidth}
    \begin{center}
        \begin{tikzpicture}[xscale=1.2]
            \node (c1) at (-1.5,5) {\machine};
            \node (c2) at (0.5,5) {\machine};
            \node at (1.2,5) {$i$};
            
            \node (g1) at (-2,3) {\job};
            \node (g2) at (-1,3) {\job};
            \node (g3) at (0,3) {\job};
            \node (g4) at (1,3) {\job};
            \node at (1.5,3) {$j$};

            \draw[black]
                (c1) -- node[left] {$b$} (g1)
                (c1) -- node[right, pos=0.2] {$\si$} (g2)
                (c2) -- node[right, pos=0.8] {$\bi$} (g2)
                (c2) -- node[right, pos=0.8] {$\si$} (g3)
                (c2) -- node[right, pos=0.2] {$\si$} (g4)
                (c2) -- node[above,pos=0.2] {$b$} (g1);
        \end{tikzpicture}
    \end{center}
    \caption{Two-value \makespan instance $I$.}
    \end{subfigure}
    \begin{subfigure}[b]{0.6\textwidth}
    \begin{center}
    \begin{tikzpicture}
    \node (c1) at (0,7.5) {\child};
    \node (g1) at (-0.8,6) {\gift};
    \node at (-1.5,6) {$r_{i'}$};
    \node (c1g1) at (0.4,6) {\gift};
    \node (c1g2) at (1.2,6) {\gift};
    \draw[black] (c1) --  node[left] {$\bi$} (g1);
    \draw[black] (c1) -- (c1g1);
    \draw[black] (c1) -- node[right] {$\si$}  (c1g2);

    \node (c2) at (4.5,7.5) {\child};
    \node at (5.2,7.5) {$q_i$};
    \node (g2) at (3,6) {\gift};
    \node (c2g1) at (4.4,6) {\gift};
    \node (c2g2) at (5.2,6) {\gift};
    \node (c2g3) at (6,6) {\gift};
    \node at (6.6,6) {$r_i^\ell$};
    \draw[black] (c2) -- node[left] {$\bi$} (g2);
    \draw[black] (c2) -- (c2g1);
    \draw[black] (c2) -- (c2g2);
    \draw[black] (c2) -- node[right] {$\si$} (c2g3);

    \node (cj1) at (-0.5,4) {\child};
    \node (cj2) at (1.5,4) {\child};
    \node (cj3) at (3.5,4) {\child};
    \node (cj4) at (5.5,4) {\child};
    \node at (6.3,4) {$\bhat{q}_j$};

    \draw[black] (cj1) -- (g1);
    \draw[black] (cj1) -- (g2);
    \draw[black] (cj2) -- (g2);
    \draw[black] (cj2) -- (c1g1);
    \draw[black] (cj2) -- (c1g2);
    \draw[black] (cj3) -- (c2g1);
    \draw[black] (cj3) -- (c2g2);
    \draw[black] (cj3) -- (c2g3);
    \draw[black] (cj4) -- (c2g1);
    \draw[black] (cj4) -- (c2g2);
    \draw[black] (cj4) -- node[right] {$\bi$}(c2g3);

    \end{tikzpicture}
    \end{center}
    \caption{Two-value \santa instance $I'$.}
    \end{subfigure}
    \caption{The construction used in \Cref{thm:unrelated-twovalue-makespan-to-santa}.}
    \label{fig:2-makespan-to-2-santa-construction}
\end{figure}
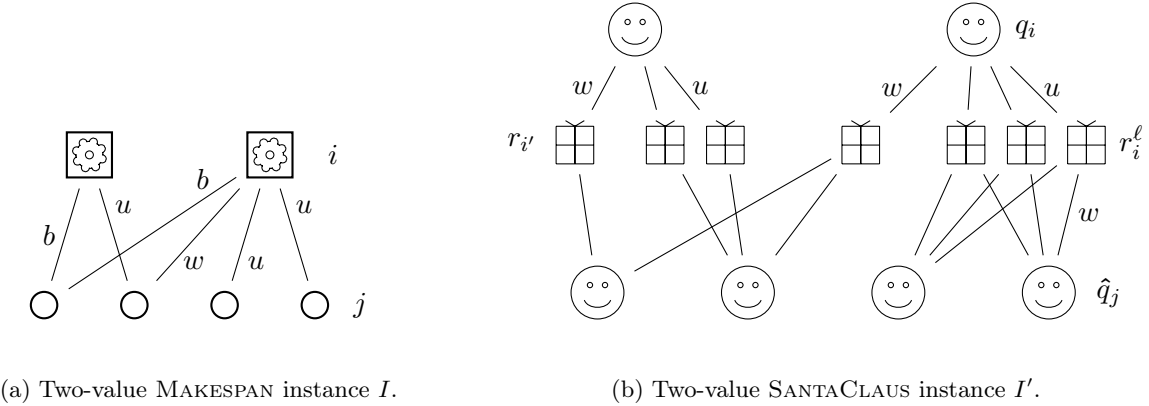

Let $t = \bi + k \cdot \si - 1$ and note that $t \le 1$ holds by construction. Also, observe that  $$t= \bi + k \cdot \si - 1 \le \bi + k\cdot \si - k \cdot \si = \bi,$$
which implies $\bi \ge t$. We first prove the following lemma.

\begin{lemma}
$\opt(I') \ge t$.
\end{lemma}

\begin{proof}
Fix an optimal solution of $I$ and recall that we assume $\opt(I) \leq 1$. In the following we construct a solution for $I'$.

Fix a machine $i$ of instance $I$. If the given solution assigns a job $j$ of size $\bi$ to $i$, we assign resource~$r_i$ to job-player $\bhat{q}_j$. If the given solution assigns a job $j$ of size $\si$ to machine $i$, we assign an arbitrary unassigned small resource $r_i^{\ell}$ to job-player $\bhat{q}_j$. All yet unassigned resources are assigned to their corresponding machine-players.

We continue by separately proving that each player in instance $I'$ receives a value of at least $t$, which implies the lemma.

First, consider the job-players. Since every job $j$ is assigned to exactly one machine in $I$, the job-player $\bhat{q}_j$ receives exactly one resource in the constructed solution for $I'$. These resources have value $\bi$ for those players, giving them a sufficiently large value of $\bi \ge t$.

Next, we consider the machine-players. Fix a machine $i$. By our assumption that~$\opt(I) \leq 1$, every machine $i$ receives jobs of total size at most $1$ in the solution for instance $I$. For our constructed solution to instance $I'$, this means that the subset $N_i$ of resources $r_i, r_i^1,\ldots, r_i^{k}$ that are \emph{not} assigned to machine-player $q_i$ satisfies $v_i(N_i) = \sum_{r \in N_i} v_{r,q_i} \le 1$. This implies that the value assigned to machine-player $i$ is a least
\[
\bi + k \cdot \si - v_i(N_i) \ge \bi + k \cdot \si - 1 = t,
\]
which implies $\opt(I') \ge t$.
\end{proof}

\begin{lemma}\label{lemma:makespan-to-santa-approx}
For any $\alpha \geq 2$, given an $\alpha$-approximate solution for $I'$, we can construct in polynomial time a solution for $I$ where every machine has a makespan of at most $2 - 1 / \alpha$.
\end{lemma}

\begin{proof}
Consider an $\alpha$-approximate solution for instance $I'$. Such a solution must assign resources of value at least $\opt(I') / \alpha$ to each player. By the previous lemma, $\opt(I') / \alpha \ge t / \alpha$.

Clearly, in the given solution for $I'$ every job-player $\bhat{q}_j$ receives at least one resource, as otherwise the objective value would be zero. We modify the given solution in a way such that each job-player receives exactly one resource. 
If a job-player receives more than one resource, we select an arbitrary resource and reassign all other resources to their corresponding machine-players. 
Now, we can construct a solution for $I'$ as follows. If a resource belonging to machine $i$ has been assigned to a job-player $\bhat{q}_j$, we assign job $j$ to machine $i$. By the above assumption this assignment is well-defined.

It remains to argue about the load of every machine in $I$. Thus, fix a machine $i$. In the solution to instance $I'$, the corresponding machine-player receives resources of value at least $t / \alpha$.
This means that the subset $N_i$ of resources $r_i, r_i^1,\ldots, r_i^{k}$ that are \emph{not} assigned to machine-player $q_i$ satisfies 
\[
v_i(N_i) \le \bi + k \cdot \si - \frac{t}{\alpha} =  1+t - \frac{t}{\alpha}.
\]
Since by construction $t \le 1$, we have $t-t/\alpha \le 1-1/\alpha$, which implies $v_i(N_i) \le 2-1/\alpha$. We conclude the proof by observing that, by construction, the makespan of machine $i$ is exactly equal to $v_i(N_i) \le 2-1/\alpha$.

\end{proof}

A visualization of the argument used in \Cref{lemma:makespan-to-santa-approx} is given in \Cref{fig:2-makespan-to-2-santa-approx}.

\begin{figure}
    \begin{subfigure}[b]{0.42\textwidth}
    \begin{center}
        \begin{tikzpicture}[xscale=1.2]
            \node (c1) at (-1.5,5) {\machine};
            \node (c2) at (0.5,5) {\machine};
            \node (g1) at (-2,3) {\job};
            \node (g2) at (-1,3) {\job};
            \node (g3) at (0,3) {\job};
            \node (g4) at (1,3) {\job};

             \draw[black]   (c1) --  (g1);
             \draw[black]   (c2) --  (g2);
             \draw[asg,col2]   (g2) -- (c1);
             \draw[asg,col1]   (g1) -- (c2);
             \draw[asg,col1]  (g3) -- (c2);
             \draw[asg,col1]  (g4) -- (c2);

             \draw [thick,col1, dotted] (-0.3,2.5) rectangle (1.3,3.5);
            \node[col1] at (2.1,4) {\small $\leq (1+t) - \frac{t}{\alpha} - \bi$};
            \node[col1] at (-0.5,4.5) {\small $\bi$};
        \end{tikzpicture}
    \end{center}
    \caption{$(2-1/\alpha)$-approximation in instance $I$.}
    \end{subfigure}
    \hfill
    \begin{subfigure}[b]{0.57\textwidth}
    \begin{center}
    \begin{tikzpicture}
    \node (c1) at (0,7.5) {\child};
    \node (g1) at (-0.8,6) {\gift};
    \node (c1g1) at (0.4,6) {\gift};
    \node (c1g2) at (1.2,6) {\gift};
    \draw[asg] (g1) -- (c1);
    \draw[asg]  (c1g1) -- (c1);
    \draw[black] (c1) --  (c1g2);

    \node (c2) at (4.5,7.5) {\child};
    \node (g2) at (3,6) {\gift};
    \node (c2g1) at (4.4,6) {\gift};
    \node (c2g2) at (5.2,6) {\gift};
    \node (c2g3) at (6,6) {\gift};
    \draw[black] (c2) -- (g2);
    \draw[black] (c2) -- (c2g1);
    \draw[black] (c2) --  (c2g2);
    \draw[asg] (c2g3) --  (c2);

    \node (cj1) at (0.5,4) {\child};
    \node (cj2) at (2.5,4) {\child};
    \node (cj3) at (4.5,4) {\child};
    \node (cj4) at (6.5,4) {\child};

    \draw[black] (cj1) -- (g1);
    \draw[black] (cj2) -- (g2);
    \draw[black] (cj2) -- (c1g1);
    \draw[black] (c2g2) -- (cj3);
    \draw[black] (c2g1) -- (cj4);
    \draw[black] (c2g3) -- (cj3);
    \draw[black] (c2g3) -- (cj4);
    \draw[asg, col2] (c1g2) -- (cj2);
    \draw[asg,col1] (c2g1) -- node[left] {\small $\bi$}  (cj3);
    \draw[asg,col1] (c2g2) -- (cj4);
    \draw[asg, col1] (g2) --  (cj1);

    \draw [thick,col1, dotted] (4,5.5) rectangle (5.5,6.5);
    \draw [thick,red, dotted] (5.6,5.5) rectangle (6.4,6.5);
    \node[col1] at (2.65,7.2) {\small $\leq  (1 + t) - \frac{t}{\alpha} - \bi$};
    \draw[dotted,col1] (3.5,7) -- (4,6.5);
    \node[red] at (7,6.5) {\small $\geq \frac{t}{\alpha}$};

    \end{tikzpicture}
    \end{center}
    \caption{$\alpha$-approximation in instance $I'$.}
    \end{subfigure}
    \caption{Visualization of the argument for translating approximate solutions used in \Cref{lemma:makespan-to-santa-approx}.}
    \label{fig:2-makespan-to-2-santa-approx}
\end{figure}
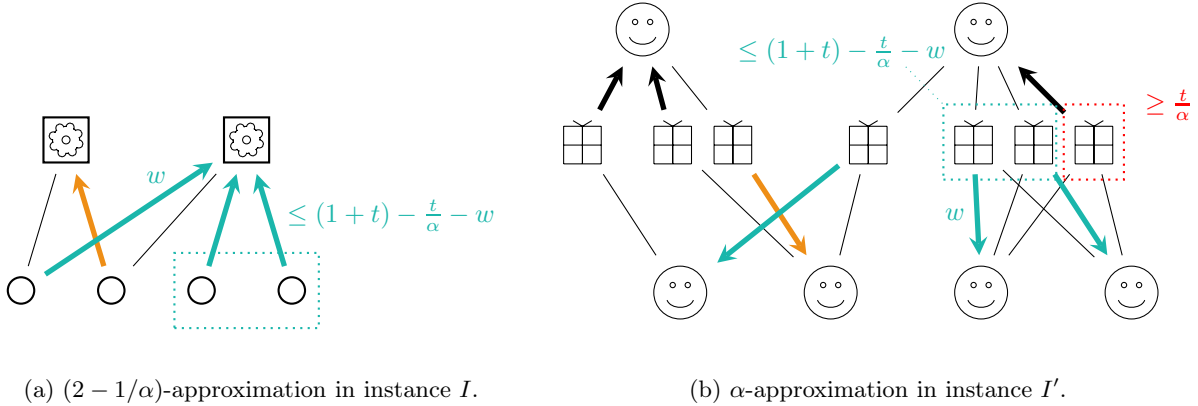

\subsection{Santa Claus to makespan}

\begin{lemma}\label{theorem:unrelated-twovalue-santa-to-makespan}
    Let $I$ be an instance of the two-value \santa problem with $\opt(I) \geq 1$. For any $\alpha \geq 2$, 
    we can construct in polynomial time an instance $I'$ of two-value \makespan such that, given an $(2-1/\alpha)$-approximate solution for $I'$, we can compute in polynomial time a solution for $I$ with an objective value of at least $1 / \alpha$.
\end{lemma}

\begin{proof}
Let $I$ be an instance of the two-value \santa problem with $\opt(I) \geq 1$ and $v_{ij} \in \{0,\si,\bi\}$. We assume w.l.o.g. that $\si \leq \bi$. We consider three exhaustive cases.

In the first case we assume that $\bi < \opt(I) / \alpha$. Then, we can use the algorithm of Bezakova and Dani~\cite{bezakova2005allocating} to compute in polynomial time a solution in which every player receives a total value of at least $\opt(I) - v_{\max} = \opt(I) - \bi > (1 - 1/\alpha) \opt(I) \geq \opt(I) / \alpha$, using $\alpha \geq 2$.

In the second case we assume that $\bi \geq \opt(I) / \alpha$ and that there is an optimal solution for $I$ in which every player $i$ receives a resource $j$ of value $v_{ij} = \bi$. Then, we can essentially set $\si$ to $0$ and compute a solution where every player receives a resource of value $\bi$ by solving a bipartite matching problem.  Since every player receives a value of at least $\bi \geq \opt(I) / \alpha \geq 1 / \alpha$, we are done. 

In the last case we assume that $\bi \geq \opt(I) / \alpha$ and that in every optimal solution for $I$ there is some player $i$ which does not receive a resource $j$ of value $v_{ij} = \bi$. We can conclude that such a player must receive at least $b = \ceil{1/\si}$ many resources $j'$ for which she has a value of $v_{ij'} = \si$, because $\opt(I) \geq 1$. 
Then, we construct a new instance $I'$ by copying $I$ and adjusting the resource values to $\bi' = 1$ and $\si' = 1 / b$. 
Observe that $\opt(I') \geq 1$ and that any solution for $I'$ in which a player either receives a resource of value $1$ or $b$ resources of value $1/ b$ gives an objective value of at least 1.
We can therefore define a collection of configurations $\cC$ in which every player has these two configurations. Then, we can use \Cref{thm:reduction-polyconfig-santa-to-makespan} (by noting that in the constructed \makespan instance every job has either size $1$ or $1/b$) to compute a solution for $I'$ in which every player receives a total value of at least $1/\alpha$.
Since $\si \geq \si'$ and $\bi \geq \opt(I) / \alpha$, this means that we can use the same solution for instance $I$ and guarantee that every player receives a total value of at least $1/\alpha$.
\end{proof}

\section{Reductions for matroid allocation problems}
\label{apx:restr-matroid}

\lemMatroidReductionSM*

\begin{proof}
Let $I$ be an instance of the restricted \matroidsanta problem with two resources with associated polymatroids $\PM_1,\PM_2$ and values $v_1,v_2 \geq 0$ such that $\opt(I) \geq 1$.
By \Cref{lemma:two-value-santa-to-general-matroid} we can assume that we are given a simplified case. For convenience,
we slightly reformulate it as follows:
given $v_1 = 1$ (instead of $\infty$), $v_2 = 1/b \leq v_1$ (instead of $1$) for $b\in\mathbb N$, and $\opt(I) \ge 1$, find a solution of value at least $1 / \alpha$.

Consider the polymatroids $\PM'_1 = \{x \in \PM_1 : x(e) \leq 1 \; \forall e \in E \}$ and $\PM'_2 = \{x \in \PM_2 : x(e) \leq b \; \forall e \in E \}$. 
Let $\dPM_1$ be the dual polymatroid of $\PM'_1$ with respect to the vector $1 \cdot E$, and let $\dPM_2$ be the dual polymatroid of $\PM'_2$ with respect to the vector $b \cdot E$.
We compose an instance $I'$ of \matroidmakespan using machines $E$, one job of size $p_1 = 1$  with polymatroid $\dPM_1$ and one job of size $p_2 = 1/b$ with polymatroid $\dPM_2$.

We now show that $\opt(I') \leq 1$. Fix an optimal solution for $I$ which selects bases $x_1 \in \cB(\PM_1)$ and $x_2 \in \cB(\PM_2)$. Since $\opt(I) \geq 1$, we have $v_1 \cdot x_1(e) + v_2 \cdot x_2(e) \geq 1$ for every $e \in E$. 
Further observe that our choice of $b$ implies $x_1(e) + (1/b) \cdot x_2(e) \geq 1$.
Our goal is to dualize bases of  $\PM'_1$ and $\PM'_2$ to obtain bases of  $\dPM_1$ and $\dPM_2$, which are feasible for $I'$. To this end, we first construct vectors $x'_1 \in \PM'_1$ and $x'_2 \in \PM'_2$ such that $x'_1(e) + (1/b) \cdot x'_2(e) \geq 1$ for all $e \in E$. This can be done by restricting $x_1$ to values of at most 1 and $x_2$ to values of at most $b$, and then selecting any bases which dominate these intermediate vectors. 
Now, we define $\overline{x}_1(e) = 1 - x'_1(e)$ and $\overline{x}_2(e) = b - x'_2(e)$ for all $e \in E$.
By the construction of $\dPM_1$ and $\dPM_2$, this solution is feasible for $I'$, i.e., $\overline{x}_j \in\cB(\dPM_j)$ for $j \in \{1,2\}$. We further have for every machine $e \in E$
\[
    \overline{x}_1(e) + \frac{1}{b} \cdot \overline{x}_2(e) = 2 -  \left(x'_1(e) + \frac{1}{b}  \cdot x'_2(e)\right) \leq 1,
\]
showing that $\opt(I') \leq 1$.

We finally prove the stated bound on the objective value of an approximate solution. Fix a $(2 - 1 / \alpha)$-approximate solution for $I'$ which selects bases $\overline{y}_1 \in \cB(\dPM_1)$ and $\overline{y}_2 \in \cB(\dPM_2)$. 
We construct an approximate solution for $I$ by defining $y'_1(e) = 1 - \overline{y}_1(e)$ and $y'_2(e) = b - \overline{y}_2(e)$ for every $e \in E$, meaning that $y'_j \in \cB(\PM'_j)$, and then choose an arbitrary basis $y_j \in \cB(\PM_j)$ which dominates $y'_j$, for $j \in \{1,2\}$.
We have for every player $e \in E$
\[
y_1(e) + \frac{1}{b} \cdot y_2(e) \geq  y'_1(e) + \frac{1}{b} \cdot y'_2(e)  = 2 - \left( \overline{y}_1(e) + \frac{1}{b} \cdot \overline{y}_2(e) \right) \geq 2 - \left( 2 - \frac{1}{\alpha} \right)  \cdot  \opt(I') \geq \frac{1}{\alpha}.
\]

Since $v_2 = 1/b$ and $v_1 = 1$,
we conclude $v_1 \cdot y_1(e) + v_2 \cdot y_2(e) \geq 1 / \alpha$ for every $e \in E$, which implies that the value of the constructed solution $y_1$ and $y_2$ is at least~$1 / \alpha$.
\end{proof}

\section{Reduction from restricted matroid Santa Claus 
to two-value case}
\label{apx:matroid-2}
Before proving the main result in this section, we start this section with a rounding theorem which can be obtained by a variant of standard arguments in the scheduling literature (see \cite{DBLP:journals/mp/ShmoysT93}). In this section, we will often refer to a standard relaxation of the problem, called the \textit{assignment LP}. In this LP, we have one variable $x_j(i)$ for each pair job/resource $j$ and player/machine $i$. In the restricted job-matroid \makespan problem, the relaxation can be written as follows.
\begin{align*}
    \sum_{j\in J} x_j(i) \cdot p_j \le T   &\ \ \ \ \forall i\in M\\
     x_j\in \mathcal B(\mathcal P_j)   &\ \ \ \ \forall j\in J\\
    x\ge 0\ , &
\end{align*}
where $T\ge 0$ is some given makespan bound. Similarly in the restricted resource-matroid \santa problem, the assignment LP becomes the following.
\begin{align*}
    \sum_{j\in R} x_j(i) \cdot v_j \ge T   &\ \ \ \ \forall i\in P\\
     x_j\in \mathcal B(\mathcal P_j)   &\ \ \ \ \forall j\in R\\
    x\ge 0\ . &
\end{align*}

\begin{theorem}
    \label{thm:rounding}
    Given a fractional assignment $x$ of resources to players which is feasible for the assignment LP (with parameter $T$) of some an instance $I$ of the restricted \matroidsanta problem, we can obtain, in polynomial time, a feasible integral assignment of resources to players of value at least $T-\max_{j\in R} v_j$.
\end{theorem}
\begin{proof}
    We assume w.l.o.g.\ that $v_1\ge v_2\ge \cdots \ge v_n$ and that $v_n=0$. We denote the polymatroids associated to the resources as $\mathcal P_1,\mathcal P_2,\dotsc,\mathcal P_n$, corresponding to the submodular functions $f_1,f_2,\ldots, f_n$. Given the fractional assignment $x$, we will create a feasible fractional solution $x'$ to a certain polymatroid intersection problem. We define the two polymatroids using a bipartite graph as follows. On the left-hand side, we have a set of vertices $W$ with one vertex $w_j$ for each resource $j$, and on the right-hand side we have a set of vertices $U$ with one vertex $u_{ij}$ for each player $i$ and resource $j$. The edge set will be denoted by $E$, and both polymatroids will have $E$ as a ground set. We set $E=\{(w_j,u_{ij})\}_{i\in P, j\in R}\cup \{(w_j,u_{i(j+1)})\}_{i\in P, j\in R\setminus \{n\}}$. For some edge $e\in E$, we denote by $e_w$ the resource corresponding to its left-hand side vertex, and by $e_u$ the player corresponding to its right-hand side endpoint.
    The first polymatroid $\PM'_1$ will be associated with the submodular function 
    \begin{equation*}
        f_1(S):=\sum_{j=1}^n f_j\left(\bigcup_{e\in S:e_w=j} e_u\right) \ .
    \end{equation*}
    The second polymatroid $\PM'_2$ will be defined using the right-hand side vertices in our graph. Each vertex $u \in U$ will have some degree constraint $d(u)$ and the submodular function $f_2$ is simply defined as
    \begin{equation*}
        f_2(S):=\sum_{e=(w,u)\in S} d(u) \ .
    \end{equation*}
    We define the degree constraints using the following process for each player $i$ in instance $I$. Let us fix one player $i$. We start with 
    \begin{equation*}
        d(u_{i1})=\left\lfloor x_1(i)\right\rfloor\ ,
    \end{equation*}
    and we define the \textit{remainder} $R_1:=x_1(i)-d(u_{i1})$ (for ease of notation we define $R_0=0$). Then we define recursively the degree constraint $d(u_{ij})$ and remainders as follows.
    \begin{align*}
        d(u_{ij}) &:=\left\lfloor R_{j-1}+x_j(i)\right\rfloor\ , \mbox{ and} \\
        R_j &:=\{R_{j-1}+x_j(i)\}= R_{j-1}+x_j(i)-\left\lfloor R_{j-1}+x_j(i)\right\rfloor = R_{j-1}+\{x_j(i)\}-\left\lfloor R_{j-1}+\{x_j(i)\}\right\rfloor\ .
    \end{align*}
    By construction, it is clear that the fractional assignment $x$ in instance $I$ can be transformed into a feasible fractional solution to our polymatroid intersection problem. To see this, fix a player $i$. Then the first resource $j=1$ can be assigned to vertex $u_{i1}$ up to an extent of $\left\lfloor x_1(i)\right\rfloor$, which is represented in our graph as taking $\left\lfloor x_1(i)\right\rfloor$ copies of the edge $(w_1,u_{i1})$. The remaining fraction of $x_1(i)$ can be assigned to player $i$ by taking the edge $(w_1,u_{i2})$ fractionally by some amount $\{x_1(i)\}$. Then we move on to job $j=2$, and we take the edge $(w_2,u_{i2})$ by the maximal amount possible until the degree constraint on vertex $u_{i2}$ becomes tight. By construction, we see that there might be some leftover of the value $x_2(i)$ which is precisely equal to the number $R_2$ in our construction. We continue this assignment until the last job $j=n$. Note that our definition of remainder $R_s$ is precisely this small leftover of $x_j(i)$ that carries over to the edge going to vertex $u_{i(j+1)}$ (note that this remainder is always less than $1$). There is one slight caveat at the end is that some small amount of the fractional assignment $x_n(i)$ might be thrown away, but as we will see, this does not hurt our purpose because we assume that $v_n=0$.

    Note that this fractional solution to our polymatroid intersection problem makes all the degree constraints on the right-hand side tight. By integrality of the polymatroid intersection polytope (see Chapters 46-47 in \cite{schrijver2003combinatorial}) there exists an integral solution which also makes all the degree constraints on the right-hand side tight (and we can find it in polynomial time by finding the maximum cardinality multiset of edges which belongs to the polymatroids intersection).
    
    It is easy to see that this integral solution corresponds to an integral assignment of resources to players in the instance $I$, in which each player $i$ receives a value of at least
    \begin{equation*}
        \sum_{j=1}^n \left\lfloor R_{j-1}+x_j(i)\right\rfloor v_j\ .
    \end{equation*}
    Let us compute the difference in objective $\Delta_i$ with the fractional solution. We have that 
    \begin{equation*}
        \Delta_i\le \sum_{j=1}^n (x_j(i)-\left\lfloor R_{j-1}+x_j(i)\right\rfloor) v_j + v_n= \sum_{j=1}^n (\{x_j(i)\}-\left\lfloor R_{j-1}+\{x_j(i)\}\right\rfloor)v_j\ .
    \end{equation*}
    Looking at each term individually, we notice that either $\left\lfloor R_{j-1}+\{x_j(i)\}\right\rfloor=0$, or that $\left\lfloor R_{j-1}+\{x_j(i)\}\right\rfloor=1$. In the first case the next remainder $R_{j}$ is equal to $R_{j-1}+\{x_j(i)\}<1$. In the second case, we have that 
    \begin{equation*}
        R_{j} = R_{j-1}+\{x_j(i)\}-\left\lfloor R_{j-1}+\{x_j(i)\}\right\rfloor = R_{j-1}+\{x_j(i)\} - 1\ .
    \end{equation*}
    Let us denote by $R'$ the set of all indices where the second case happens. Then we can write
    \begin{equation*}
        \Delta_i\le \sum_{j=1}^n \{x_j(i)\} v_j - \sum_{j\in R'} v_j\ .
    \end{equation*}

    Now note that if $j\notin R'$, then $\{x_j(i)\}=R_j-R_{j-1}$, and that if $j\in R'$ then $\{x_j(i)\}=1+R_j-R_{j-1}$. Using this observation, we obtain 
    \begin{multline*}
        \Delta_i\le \sum_{j=1}^n (R_{j}-R_{j-1}) v_j + \sum_{j\in R'}v_j- \sum_{j\in R'} v_j\le  \sum_{j=1}^n (R_{j}-R_{j-1}) v_j = \sum_{j=1}^{n-1} (v_{j}-v_{j+1})R_j \le v_1  \ ,
    \end{multline*}
    where we use the fact that $R_0=v_n=0$, and $R_j\le 1$ for all $j$.
\end{proof}

\begin{lemma}\label{lemma:two-value-santa-to-general-matroid}
    For any $\alpha \ge 2$,
    if there is a polynomial-time algorithm that, given an instance of restricted two-value \matroidsanta problem with one matroid of value $v_1 = \infty$ and one polymatroid of value $v_2 = 1$ and a number $b\in\mathbb N$, finds a solution of value at least $b$ or determines that there is no solution of value $\alpha b$,
    then there is also:
    \begin{enumerate}
        \item a polynomial-time $\alpha$-approximation algorithm for (any instance of) the restricted two-value \matroidsanta problem and
        \item a polynomial-time $2\alpha$-approximation algorithm for the restricted \matroidsanta problem.
    \end{enumerate}
\end{lemma}

\begin{proof}
    We start by proving the first result of the lemma. Let $\si,\bi$ (with $\bi\ge \si$) be the two sizes of the instance $I$, and let $f_\si,f_\bi$ be the associated submodular functions.

    Now, we have three possible cases. If $\opt(I)/\alpha \le \si$, then it suffices to give at least one resource to any player to obtain an $\alpha$-approximate solution. This can be checked in polynomial time (see Chapter 42 in \cite{schrijver2003combinatorial}). If $\si < \opt(I)/\alpha\le \bi$, then, in an $\alpha$-approximate solution, it suffices to give to each player either one resource of value $\bi$ or $\opt(I)/(\alpha\si)$ resources of value $\si$. We define a new instance $I_2$ with one matroid of value $v_1=\infty$ and one polymatroid of value $v_2=1$. The independent sets of the matroid are the sets of players which can be covered by at least one resource of value $\bi$ each in the original instance. The polymatroid is defined by the submodular function $f_2(S):= f_\si(S)$. Clearly, in this case, we have that $\opt(I_1)\ge \opt(I)/\si$, and a solution of value $t$ in instance $I_1$ can immediately be translated into a solution of value $\min\{\opt(I)/\alpha,t \si\}$ in the original instance. So an $\alpha$-approximation on instance $I_1$ gives us an $\alpha$-approximation on the original instance. In the last case where $\opt(I)/\alpha > \bi$, we
    first assume without loss of generality that $u, w$ are integers (by appropriate scaling).
    Then we define a polymatroid $\PM_3$ with the submodular function $f_3(S):=\si\cdot f_\si(S)+\bi\cdot f_\bi(S)$; this should be thought of splitting the resources of value $\bi$ (respectively $\si$) into $\bi$ (respectively $\si$) individual resources of value $1$ each. The biggest $b$ such that $b\cdot E\in \PM_3$ (where $E$ is the set of all players and $b\cdot E$ is the $|E|$ dimensional vector with all entries equal to $b$) can be found in polynomial time, since we simply need to minimize a submodular function (see Chapter 45 in \cite{schrijver2003combinatorial}). This gives us a fractional solution to the original instance of objective value $\opt(I)$. Using Theorem~\ref{thm:rounding}, we can round (in polynomial time) this fractional solution into an integral solution of objective value $\opt(I)-\max\{\si,\bi\}\ge \opt(I)\cdot (1-1/\alpha)\ge \opt(I)/\alpha$ (using $\alpha \ge 2$), which concludes the proof of the first point of the lemma.

    For the second point, by a standard guessing strategy (as explained in the beginning of Section $\ref{sec:santa-to-makespan}$), we can assume that we know the optimum value $\opt(I)$. We call a resource $j$ \textit{heavy} if $v_j\ge \opt(I)/(2\alpha) $, and \textit{light} otherwise (let $H$ and $L$ be the set of heavy and light resources respectively). We then define an instance $I'$ with one matroid of value $\infty$, whose independent sets are the sets of players which can be covered by at least one heavy resource each (this is a matroid by the matroid union theorem see Chapter 42 in \cite{schrijver2003combinatorial}). Again, assuming that all values $v_j$ are integers, we define one polymatroid of value $1$ associated to the submodular function $f'(S):=\sum_{j\in L} v_jf_j(S)$. In this new instance, it is clear that $\opt(I')\ge \opt(I)$. Hence an $\alpha$-approximate solution to instance $I'$ can be transformed into an $\alpha$-approximate solution to instance $I$, in which the heavy resources are assigned integrally, and the light resources fractionally. Using Theorem~\ref{thm:rounding}, we can round this fractional assignment into and integral assignment of value at least $\opt(I)/\alpha-\max_{j\in L}v_j\ge \opt(I)/\alpha-\opt(I)/(2\alpha)=\opt(I)/(2\alpha)$.
\end{proof}

\begin{theorem}
    \label{thm:rounding2}
     Given a fractional assignment $x$ of jobs to machines which is feasible for the assignment LP (with parameter $T$) of some an instance $I$ of the restricted \makespan problem, we can obtain, in polynomial time, a feasible integral assignment of jobs to machines of makespan at most $T+\max_{j\in J} p_j$. This implies that we can find in polynomial time a feasible integral solution of makespan at most $\opt(I)+\max_{j\in J} p_j \le 2\opt(I).$
\end{theorem}
\begin{proof}
    The proof here is very similar to the proof of Theorem \ref{thm:rounding}. We repeat it here for completeness.

    We assume w.l.o.g.\ that $p_1\ge p_2\ge \cdots \ge p_n$ and that $p_n=0$. We denote the polymatroids associated to the jobs as $\mathcal P_1,\mathcal P_2,\dotsc,\mathcal P_n$, corresponding to the submodular functions $f_1,f_2,\ldots,f_n$. Given the fractional assignment $x$, we will create a feasible fractional solution $x'$ to a certain polymatroid intersection problem. We define the two polymatroids using a bipartite graph as follows. On the left-hand side, we have a set of vertices $W$ with one vertex $w_j$ for each job $j$, and on the right-hand side we have a set of vertices $U$ with one vertex $u_{ij}$ for each machine $i$ and job $j$. The edge set will be denoted by $E$, and both polymatroids will have $E$ as a ground set. We set $E=\{(w_j,u_{ij})\}_{i\in M, j\in J}\cup \{(w_j,u_{i(j-1)})\}_{i\in M, j\in J\setminus \{1\}}$ (note the slight change here compared to Theorem \ref{thm:rounding}). For some edge $e\in E$, we denote by $e_w$ the job corresponding to its left-hand side vertex, and by $e_u$ the machine corresponding to its right-hand side endpoint.
    The first polymatroid $\PM'_1$ will be associated with the submodular function 
    \begin{equation*}
        f_1(S):=\sum_{j=1}^n f_j\left(\bigcup_{e\in S:e_w=j} e_u\right) \ .
    \end{equation*}
    The second polymatroid $\PM'_2$ will be defined using the right-hand side vertices in our graph. Each vertex $u \in U$ will have some degree constraint $d(u)$ and the submodular function $f_2$ is simply defined as
    \begin{equation*}
        f_2(S):=\sum_{e=(w,u)\in S} d(u) \ .
    \end{equation*}
    We define the degree constraints using the following process for each machine $i$ in instance $I$. We start with 
    \begin{equation*}
        d(u_{i1})=\left\lceil x_1(i)\right\rceil\ ,
    \end{equation*}
    and we define the \textit{remainder} $R_1:=d(u_{i1})-x_1(i)$. Then we define recursively the degree constraint $d(u_{ij})$ ($j\ge 2$) and remainders as follows.
    \begin{align*}
        d(u_{ij}) &:=\left\lceil x_j(i)-R_{j-1}\right\rceil\ , \mbox{ and} \\
        R_j &:=\left\lceil x_j(i)- R_{j-1}\right\rceil - (x_j(i)-R_{j-1}) = \left\lceil \{x_j(i)\}- R_{j-1}\right\rceil - (\{x_j(i)\}-R_{j-1})\ .
    \end{align*}
    It is easy to see that the solution $x$ to the assignment LP can be transformed into a feasible fractional solution to our polymatroid intersection problem. Similar to the proof of Theorem \ref{thm:rounding}, the remainder $R_j$ is defined to be exactly the quantity by which we can select the edge $(w_j,u_{i(j-1)})$ in our fractional solution to the polymatroid intersection problem.

    Note that this fractional solution to our polymatroid intersection problem is such that we have a basis of the polymatroid $\mathcal P'_1$ (the left-hand side polymatroid). By integrality of the polymatroid intersection polytope (see Chapters 46-47 in \cite{schrijver2003combinatorial}) there exists an integral solution which is also a basis in $\mathcal P'_1$ (and we can find it in polynomial time by finding the maximum cardinality multiset of edges which belongs to the polymatroids intersection).
    
    It is easy to see that this integral solution corresponds to an integral assignment of jobs to machines in the instance $I$, in which each machine $i$ receives a load of at most
    \begin{equation*}
        \sum_{j=1}^n \left\lceil x_j(i)-R_{j-1}\right\rceil p_j \ .
    \end{equation*}
    Let us compute the difference in objective $\Delta_i$ with the fractional solution. We have that 
    \begin{multline*}
        \Delta_i\le \sum_{j=1}^n (\left\lceil x_j(i)-R_{j-1}\right\rceil -x_j(i)) p_j= \sum_{j=1}^{n-1} (\left\lceil x_j(i)-R_{j-1}\right\rceil -x_j(i))p_j \\
        = \sum_{j=1}^{n-1} (\left\lceil \{x_j(i)\}-R_{j-1}\right\rceil -\{x_j(i)\})p_j\ ,
    \end{multline*}
    using $p_n=0$. Looking at each term individually, we notice that either $\left\lceil \{x_j(i)\}-R_{j-1}\right\rceil=0$, or that $\left\lceil \{x_j(i)\}-R_{j-1}\right\rceil=1$. In the first case the next remainder $R_{j}$ is equal to $R_{j-1}-\{x_j(i)\}$. In the second case, we have that 
    \begin{equation*}
        R_{j} = 1+R_{j-1}-\{x_j(i)\} \iff \{x_j(i)\} =  1+R_{j-1}-R_j\ .
    \end{equation*}
    Let us denote by $J'$ the set of all indices where the second case happens. Then we can write
    \begin{multline*}
        \Delta_i\le \sum_{j\in J'} p_j - \sum_{j=1}^n \{x_j(i)\}p_j = \sum_{j\in J'} p_j-\sum_{j\in J'} p_j - \sum_{j=1}^n (R_{j-1}-R_j)p_j = \sum_{j=1}^n (R_{j}-R_{j-1})p_j\\
        =  \sum_{j=1}^{n-1} (p_{j}-p_{j+1})R_j \le p_1\ ,
    \end{multline*}
    where we use the fact that $R_0=p_n=0$, and $R_j\le 1$ for all $j$.
\end{proof}

\section{Analysis of local search algorithm}
\label{apx:ana-localsearch}

\subsection{General properties of matroids and submodular functions}
\begin{lemma}\label{lem:decr-marg}
    Let $g: E \rightarrow \mathbb Z_{\ge 0}$ be monotone submodular with $g(\emptyset) = 0$, $X' \subseteq X\subseteq E$, and $h' < h$.
    Then 
    \begin{align*}
      g(Y \mid h \cdot X') &\ge g(Y \mid h \cdot X) \text{ for all } Y\subseteq E\setminus X \text{ and} \\
      g(Y \mid h' \cdot X) &\ge g(Y \mid h \cdot X) \text{ for all } Y\subseteq E\setminus X \ . \\
    \end{align*}
\end{lemma}
\begin{proof}
    Recall that $g(Y \mid h\cdot X)$ is derived by first 
    constructing a new monotone submodular function $g'$ corresponding to the
    polymatroid defined by $g$ but with entries of $X$ bounded
    by $h$. 
    For the first inequality,
    let $g''$ be the corresponding function with $X'$ instead of $X$.
    Then
    \begin{multline*}
        g(Y \mid h \cdot X') = g''(Y \mid X')
        = g''(Y \cup X') - g''(X') \\
        \ge g'(Y \cup X') - g'(X') = g'(Y \mid X') \ge g'(Y \mid X) = g(Y \mid h \cdot X) \ . 
    \end{multline*}
    Here we use that $g''(X') = g'(X')$.
    For the second inequality,
    let $g'''$ be the submodular function for $h'$ instead of $h$. Let $\mathcal Q$ be the polymatroid corresponding to $g$. Then
    there exists some $y'''\in \mathcal Q$ with $\mathrm{supp}(y''') \subseteq X$, $y'''(i)\le h'$ for all $i\in E$, and $y'''(X) = g'''(X)$. We define $y'$ accordingly, except for $g'$ instead of $g'''$. Here we may assume, by augmentation property of the polymatroid that $y'(i)\ge y'''(i)$ for all $i\in E$. 
    The value $g(Y \mid h' \cdot X) = g'''(Y \mid X)$
    is simply the largest element (by sum) in the polymatroid defined by restricting $\mathcal Q$ to $E\setminus X$ and replacing the submodular function by $g'''(Y') = g(Y' \cup X) - y'''_i(X)$.
    It is clear that this is at least as big as
    $g(Y \mid h \cdot X) = g'(Y \mid X)$, since
    the corresponding polymatroid here has a submodular function $g'(Y') = g(Y' \cup X) - y'_i(X)$ that is smaller or equal to $g'''$ everywhere.
\end{proof}

\begin{lemma}\label{lem:less-d}
    Let $g: E \rightarrow \mathbb Z_{\ge 0}$ be monotone
    submodular with $g(\emptyset) = 0$ and $X \subseteq E$.
    Define $Y = \{i\in X : g(i \mid h \cdot (X - i)) < h\}$.
    Then for every $i\in X$ we have $g(i\mid h \cdot (Y - i)) < h$ if and only if $i\in Y$.
\end{lemma}
\begin{proof}
    For any $i\in X$ with $g(i \mid h \cdot (Y - i)) < h$ we have $g(i \mid h \cdot (X - i)) < h$ by Lemma~\ref{lem:decr-marg}. This proves one direction. For the other direction,
    let $i\in X$ with $g(i \mid h \cdot (Y - i)) \ge h$.
    Further, let $\mathcal Q$ be the polymatroid corresponding to $g$ and
    let $y\in \mathcal Q$ with $\mathrm{supp}(y) \subseteq Y - i$ and $y(j) \le h$ for all $j\in E$.
    Further, choose $y$ such that $y(Y-i)$ is maximized.
    Since $g(i \mid h \cdot (Y - i)) \ge h$, we can extend $y$ to $y'$ which is equal for all $j\neq i$ and has $y'(i) = h$.
    Now for any $i' \in X \setminus Y$,
    since $g(i' \mid h\cdot Y) \ge g(i' \mid h \cdot (X - i')) \ge h$ we can repeat the same trick and increase the value of $y'(i')$ to $h$ as well. It is easy
    to see that this can be continued to obtain
    $y''\in \mathcal Q$ with $y''(j) = h$ for all $j\in X\setminus (Y - i)$ and $y''(j) = y(j)$ for all $j\in Y - i$.

    It is clear that $y''$ when restricted to $X - i$
    maximizes $y''(X - i)$ over all elements of $\mathcal Q$
    with support $X - i$ and upper bound $h$:
    It maximizes the sum already on $Y - i$ and all
    other components are obviously the largest possible.
    Together with the fact that $y''$ with $y''(i) = h$ is in $\mathcal Q$, this implies that $g(i \mid h \cdot (X - i)) \ge h$.
\end{proof}

\begin{lemma}\label{lem:large-subm-value}
    Let $g: E\rightarrow \mathbb Z_{\ge 0}$ be monotone submodular with $g(\emptyset) = 0$ and $X'\subseteq X\subseteq E$. Further,
    let $g(i \mid h \cdot (X - i)) < h$ for all $i\in X'$.
    Then $g(X') \le h |X|$ and strict inequality holds whenever $X'\neq \emptyset$.
\end{lemma}
\begin{proof}
    We use an induction over $|X'|$. For $X' = \emptyset$ the claim obviously holds.
    Now consider $X' \neq \emptyset$ and let $i\in X'$. There must be some $Y\subseteq X$ with $i\in Y$
    such that $g(Y) < h |Y|$: assume otherwise and let $\mathcal Q$ be the polymatroid corresponding to $g$.
    Let $z\in \mathcal Q$ with $z(j) \le h$ for all $j\in E$ and $z(i) = 0$, maximizing $z(E)$.
    It can easily be checked that $z'$ with $z'(j) = z(j)$ for $j\neq i$ and $z'(i) = h$ is also in $\mathcal Q$.
    This however implies that $g(i \mid h \cdot (X - i)) \ge h$.

    Having established that $g(Y) < h |Y|$ for some $Y\ni i$, we use the induction hypothesis
    on $X\setminus Y$, $X'\setminus Y$ and $g'(Y') := g(Y' \mid Y)$. It follows that $g(X'\setminus Y \mid Y) \le h |X\setminus Y|$ and therefore
    \begin{equation*}
        g(X') \le g(Y) + g(X'\setminus Y \mid Y) < h |Y| + h |X\setminus Y| = h |X| \ . \qedhere
    \end{equation*}
\end{proof}

\subsection{Basic properties and invariants of the data structures}
Note that after $A, C, B$, and $A_I$ are
initially created, we never change $A, C, I_M\cap C$ or $I_P\cap C$.
The only dynamic sets are $B$, $A_I$, $I_M\setminus C$, and $I_P\setminus C$.
Hence, for properties that rely solely on the fixed elements, it suffices
to verify them at the time they were created.
Furthermore, the only place that makes potentially
dangerous changes to $B$, $I_M\setminus C$, and $I_P\setminus C$ is the recursion.
In the remainder, ``at all times'' means that a property should hold between any two of the four main operations that are performed repeatedly.

We will now verify the properties of the set of addable elements.
Notice that Properties~(1) and~(2), that is, $2b \cdot A \in \PM$
and $f(i \mid 2b \cdot A) < 2b$ for all $i\in C\setminus A$,
hold trivially by construction.
\begin{lemma}\label{lem:addable4}
    For the set of addable elements $A$ and set $C$ we have that
    $r(B_0 \mid C) \le 2\epsilon |B_0|$.
\end{lemma}
\begin{proof}
    Let $A_1,A_2,\dotsc,A_{\ell}$ be the sets
    created by the procedure.
    Consider
    the time that $C$ is created. Here,
    it holds that $r(B_0 \mid I_M) < \epsilon^2 |B_0| \le \epsilon |B_0|$.
    Now assume towards contradiction that $r(B_0 \mid C) > 2 \epsilon |B_0|$.
    Thus, we can find some $X \subseteq B_0$ with $|X| = r(B_0 \mid C)$ and $C \cup X \in \cI$.
    After finalizing $A$, for all $i\in I_M\setminus (C \cup A_{\ell})$ we have
    $r(i \mid B_0 \cup I_M\setminus A_{\ell} - i) = 1$.
    Thus, $X \cup I_M\setminus A_{\ell} \in \cI$, which can be seen by adding to $C \cup X$ the
    elements from $I_M\setminus (A_{\ell} \cup C)$ one at a time (each having
    marginal value $1$). Applying the matroid augmentation property on $I_M$,
    we can find a set $Y\subseteq B_0$ such that $Y\cup I_M\in \cI$ and
    \begin{equation*}
        |Y| = |X \cup I_M\setminus A_{\ell}| - |I_M| \ge |X| - |A_{\ell}| > \epsilon |B_0| \ ,
    \end{equation*}
    a contradiction.
\end{proof}
\begin{lemma}\label{lem:addable3}
    For the set of addable elements $A$ and set $C$ we have that
    $r(B_0 \mid I_M \setminus R) \ge \epsilon^2 |B_0| - |B_0 \cap I_M|$
    for every $R\subseteq A$ with $|R| \ge \epsilon |A|$.
\end{lemma}
\begin{proof}
    We will first prove the statement for
    the time when $A$ was created, but then we need
    to show that it also holds later.
    Let $A_1,A_2,\dotsc,A_{\ell}$ be the sets
    created by the procedure and recall that $A = A_1\cup \cdots \cup A_{\ell - 1}$ and $A_{j} \ge \epsilon B_0$ for all $j < \ell$.
    Thus, there must be some $A_j$, $j < \ell$, with $|R\cap A_j| \ge \epsilon |A_j| \ge \epsilon^2 |B_0|$.
    For $I_M$ at the time of construction it
    holds that
    \begin{equation*}
    r(B_0 \cup I_M\setminus (R \cap A_j)) \ge r(B_0 \cup I_M\setminus A_j) = r(B_0 \cup I_M) \ge |I_M| \ .
    \end{equation*}
    This is because $A_j$ is constructed greedily from elements that do not decrease the rank. Hence,
    \begin{multline*}
        r(B_0 \mid I_M \setminus R) \ge r(B_0 \mid I_M \setminus (R \cap A_j)) \\
        = r(B_0 \cup I_M \setminus (R \cap A_j)) - r(I_M\setminus (R \cap A_j))
        \ge |I_M| - |I_M| + |R \cap A_j| \ge \epsilon^2 |B_0| \ .
    \end{multline*}
    We will now study the effects of $I_M$ changing
    throughout the algorithm.
    Essentially, when an element of $B_0$ is added to $I_M$ then $r(B_0 \mid I_M\setminus R)$ may decrease by $1$, otherwise it does not change.
    Note that we can
    view the changes made by the algorithm (or its recursive calls) to $I_M$ as a sequence of single insertions or deletions.
    More precisely, there exists a sequence of sets
    Let $I_1,I_2,\dotsc,I_{k} \in \cI$,
    where $I_{k}$ is the current state of $I_M$ that we want to analyze,
    $I_1$ is the initial state, and
    $I_{h+1}$ is derived from $I_h$ either by
    deletion or addition of a single element.
    Further, whenever $I_{h+1} = I_{h} + i$ for some $i\notin I_h$, then we know that
    $I_{h} + j \notin \cI$ for all $j\notin I_h$ with $j\prec i$. Finally, once an element of $B_0$ is
    added to some $I_h$, it remains in $I_M$, i.e., it is also in $I_{h+1}, \dotsc, I_{k}$. All of these properties are observations that
    follow easily from the definition of the algorithm.

    Let $1 \le h < k$, $R\subseteq C \subseteq I_h, S\subseteq E\setminus I_h$ such that $|R| \ge |S|$ and $I_h \setminus R \cup S \in \cI$. Assume further that $I_h + s\notin \cI$ for all $s\in S$. Then $I_{h+1} \setminus R \cup S \in \cI$ as well: if $I_{h+1}$ is derived by deletion of an element, this follows immediately. Now assume that $I_{h+1} = I_h + i$.
    Since $I_h + i \in \cI$ and $I_h \setminus R \cup S \in \cI$, by matroid augmentation property there
    exists some $j\in (I_h + i) \setminus (I_h \setminus R\cup S) = R + i$ with $I_h \setminus R \cup S + j \in \cI$. If $j = i$ we are done, otherwise we get a contradiction:
    suppose that $I_h \setminus (R - j) \cup S\in \cI$.
    Then we can apply the matroid augmentation property
    to $I_h$ to find some $s\in S$ with $I_h + s\in \cI$.

    Now consider a subsequence $I_h,I_{h+1},\dotsc,I_{g}$ where no element of $B_0$ is added. Then if $I_h \setminus R\cup S\in \cI$ it follows that $I_g \setminus R\cup S\in \cI$. Notice that as we have shown earlier, for every $R\subseteq A \subseteq C$ with $|R| \ge \epsilon |A|$ there exists some $S_1\subseteq B_0$ such that $|S_1| \ge \epsilon^2 |B_0|$ and $I_1 \setminus R \cup S_1 \in \cI$. Let $I_h$ be the first time that an element of $B_0$ is inserted into $I_M$. Then by previous arguments $I_{h-1} \setminus R \cup S_1 \in \cI$. Since $I_h$ extends $I_{h-1}$ by only one element,
    by matroid augmentation property
    $I_{h} \setminus R \cup S_2\in \cI$
    for some $S_2\subseteq S_1$ with $|S_2| = |S_1| - 1$.
    We can repeat this argument and since only $|I_k \cap B_0|$ many times
    an element from $B_0$ is added, we will finally obtain a set $S_k$
    with $|S_k| = |S_1| - |B_0 \cap I_k| \ge \epsilon^2 |B_0| - |B_0 \cap I_k|$
    and $I_k \setminus R \cup S_k\in \cI$; thus proving the lemma.
    
\end{proof}

\begin{lemma}\label{lem:inv-feasible}
At all times, $I_M$ and $I_P$ are disjoint, $I_M\in \cI$, and $b \cdot I_P\in \PM$.
\end{lemma}
\begin{proof}
The only modifications to $I_M$
are $I_M \gets C \cup I_M''$ through recursion,
where $I_M''$ is independent in the contracted matroid $\cM \slash C$ and greedy additions of elements before terminating. Both of these
operations clearly maintain $I_M\in \cI$.

For $I_P$, we will argue the stronger statement that
at all times $b \cdot (I_P \cup A_I) \in \PM$, which also implies
that the operation of adding elements from $A_I$ to $I_P$ will maintain $I_P \in \PM$.
The property that $b \cdot (I_P \cup A_I) \in \PM$ is by definition of the algorithm maintained when adding elements to $A_I$.
Consider now the operation $I_P \gets (B\setminus I_M'') \cup I_P''$ performed by the recursion, where
each element $i\in I_P''$ satisfies $f(i\mid b \cdot (I_P''\cup A \cup B - i)) \ge b$. We assume that before the recursive call we have
$b \cdot (I_P \cup A_I) \in \PM$ and, in particular, $b \cdot (B \cup A_I) \in \PM$.
Thus, because of the marginal values of all elements in $I''_P$, it follows immediately that
$b \cdot (A_I \cup B \cup I_P'') \in \PM$ and, in particular,
$b \cdot (A_I \cup (B\setminus I_M'')  \cup I_P'') \in \PM$.
Note that $(B\setminus I_M'')  \cup I_P''$ is equal to $I_P$ after the
recursion.
Since these are the only places where $I_P$ is changed, this completes the proof.
\end{proof}

\begin{lemma}\label{lem:recursive-feasible}
    The input $E', B'_0, I_M', I_P', \cM', \PM'$ created for the recursion is feasible.
    More concretely,
    \begin{enumerate}
        \item $B'_0, I_M', I_P' \subseteq E'$ are disjoint,
        \item $I_M'\in \cI'$, and
        \item $b \cdot I_P'\in \PM'$.
    \end{enumerate}
\end{lemma}
\begin{proof}
 The only non-obvious statement is $b \cdot I_P' \in \PM'$. Here, notice that $f'(X) = f(X \mid b \cdot (A\cup B))$ and for each $i\in I'_P$ we have $f(i \mid b \cdot (A\cup I_P - i) \ge b$, since $i\notin B$.
 Starting with $X = \emptyset$ and adding each element of $I'_P$ one at a time, it is easy to see that the marginal values $f'(i \mid b\cdot X)$ are always least $b$.
\end{proof}

\begin{lemma}\label{lem:ABmarg}
    For all $i\in A\cup B\setminus A_I$ we have
    $f(i \mid b \cdot (A\cup B - i)) < b$.
\end{lemma}
\begin{proof}
    Consider the time $B$ was last updated.
    By definition we have
    $f(i\mid b\cdot (A\cup I_P - i)) < b$ if and only if $i\in B$ for all $i\in I_P$. Further,
    for every $i\in A\setminus A_I$ we have $f(i\mid b \cdot (A \cup I_P - i)) \le f(i\mid b \cdot I_P) < b$.

    Let $X = \{i\in A \cup I_P : f(i \mid b\cdot (A \cup I_P - i) < b\}$.
    Then by the previous observations, $A \cup B \setminus A_I \subseteq X \subseteq A\cup B$.
    By Lemma~\ref{lem:less-d} it follows that
    $f(i \mid b \cdot (A\cup B - i)) \le f(i \mid b \cdot (X - i)) < b$ for all $i\in A\cup B\setminus A_I$.
\end{proof}
We will now prove that a recursion significantly decreases the number of blocking elements.
\begin{lemma}\label{lem:decr-B}
Let $B'$ be the set of blocking elements after a recursion has returned and assume that the algorithm
did not immediately terminate. Denote
by $B$ the blocking elements before the recursion. Then $B'\subseteq B$.
Moreover, $|B'| \le (1 - \epsilon^2) |B|$.
\end{lemma}
\begin{proof}
    Let $i\in I_P'\setminus B$, where again the prime denotes the state after the recursion. Then $f(i \mid b \cdot (A \cup I_P' - i)) \ge b$ (using the definition of the submodular function $f'$ of the recursion).
    Therefore, $i\notin B'$.
    Since $|I_M' \cap B| \ge \epsilon^2 |B|$ and such elements will not appear in $I_P'$ it follows immediately that
    $|B'| \le (1 - \epsilon^2) |B|$.
\end{proof}

\begin{lemma}\label{lem:size-B}
    At all times we have
    $|B| > (1 - 2\epsilon)|A|$
\end{lemma}
\begin{proof}
    Recall that for all $i\in A\cup B$ we have $f(i \mid b \cdot (A\cup B - i)) < b$ unless $i\in A_I$, see Lemma~\ref{lem:ABmarg}.
    Thus, from Lemma~\ref{lem:large-subm-value} it follows that $f((A\setminus A_I) \cup B) < b |A\cup B|$.
    Furthermore, $2b\cdot A\in \PM$ and $|A_I| < \epsilon |A|$, which implies that $f((A\setminus A_I) \cup B) \ge f(A\setminus A_I) \ge 2b |A\setminus A_I| \ge (2 - 2\epsilon) b |A|$.
    Putting both inequalities together, we get
    $|A\cup B| > (2 - 2\epsilon) |A|$, which simplifies to
    $|B| > (1 - 2\epsilon)|A|$.
\end{proof}

\subsection{Termination with failure}
In the case that our algorithm returns failure, we need to prove that there
does not exist a (slightly stronger) solution.
This proof comes in the form of a certificate that we will define here.
\begin{definition}{\rm
 A \emph{certificate of infeasibility} consists of two (possibly empty) sets $Z_2 \subseteq Z_1 \subseteq E\setminus B_0$ with
\begin{enumerate}
    \item $r(B_0 \mid Z_1) < 2 \epsilon |B_0|$,
    \item $r(Z_1) \le |Z_1| - (0.5 - 2\epsilon) |Z_2| + \epsilon |B_0|$,
    \item $f(i \mid b \cdot (Z_2 - i)) < b$ for at least a $(1 - \epsilon)$ fraction of elements in $Z_2$,\label{eq:margins}
    \item $f(i \mid 2b \cdot Z_2) < 2b$ for all $i\in Z_1 \setminus Z_2$
\end{enumerate}}
\end{definition}
The intuition for the certificate is that the first property states that a significant amount of $Z_1$ cannot be covered by the polymatroid if we want to cover many elements of $B_0$.
However, the other properties imply that this is not possible.
\begin{lemma}
    If there exists a certificate of infeasibility then there cannot be
    two sets $I_M^* \cup I_P^* \supseteq E\setminus B_0$ with $I_M^*\in \cI$, $\alpha b \cdot I_P^* \in \PM$,
    and $|B_0 \cap I_M^*| \ge 3\epsilon |B_0|$, where
    $\alpha = 4 + \cO(\epsilon)$.
\end{lemma}
\begin{proof}
    Let $I_M^* \cup I_P^* \supseteq E\setminus B_0$ with $I_M^*\in \cI$, $\alpha b \cdot I_P^* \in \PM$.
Let $Z'_2 \subseteq Z_2$ be the elements with $f(i\mid b\cdot (Z_2 - i)) \ge b$.
Then by~\eqref{eq:margins} we have $|Z'_2| \le \epsilon \cdot |Z_2|$ and from Lemma~\ref{lem:large-subm-value} it follows that $f(Z_2 \setminus Z'_2) \le b \cdot |Z_2|$.
We first bound
\begin{align*}
    \alpha b |I_P^* \cap (Z_1 \setminus Z_2')|
    &\le f(I_P^* \cap (Z_1 \setminus Z_2')) \\
    &\le f(Z_2 \setminus Z_2') + f(I^*_P \cap (Z_1 \setminus Z_2) \mid Z_2 \setminus Z_2') \\
    &\le f(Z_2 \setminus Z_2') + f(I^*_P \cap (Z_1 \setminus Z_2) \mid 2b \cdot (Z_2 \setminus Z_2')) \\
    &\le f(Z_2 \setminus Z_2') + 2 b |Z'_2| + f(I^*_P \cap (Z_1 \setminus Z_2) \mid 2b \cdot Z_2) \\
    &\le f(Z_2 \setminus Z_2') + 2 b |Z'_2| + \sum_{i\in I^*_P \cap (Z_1 \setminus Z_2)} f(i \mid 2b \cdot Z_2) \\
    &\le f(Z_2 \setminus Z_2') + 2 b |Z'_2| + \sum_{i\in I^*_P \cap (Z_1 \setminus Z'_2)} f(i \mid 2b \cdot Z_2) \ .
\end{align*}
From this it follows that
\begin{equation*}
  b (\alpha - 2) |(Z_1 \setminus Z'_2) \cap I_P^*| \le \sum_{i\in I_P^* \cap (Z_1 \setminus Z'_2)} (\alpha b - f(i \mid 2b \cdot (Z_2 - i))) \le 2\epsilon  b |Z_2| + f(Z_2 \setminus Z'_2) \le (1 + 2\epsilon) b |Z_2| \ .
\end{equation*}
Consequently, $|(Z_1 \setminus Z'_2) \cap I_P^*| \le (1 + 2\epsilon) / (\alpha - 2) \cdot |Z_2|$.
Thus,
\begin{equation*}
    |Z_1 \cap I_M^*| \ge |Z_1| - |(Z_1 \setminus Z'_2)\cap I_P^*| - |Z'_2| \ge |Z_1| - \frac{1 + 2\epsilon + (\alpha - 2)\epsilon}{\alpha - 2} |Z_2| \ge r(Z_1) - \epsilon |B_0| \ ,
\end{equation*}
where the last inequality holds because of Property~(2) for $\alpha = 4 + \cO(\epsilon)$ and $\epsilon$ sufficiently small.
We conclude
\begin{multline*}
    |B_0 \cap I_M^*| \le r(B_0 \mid I_M^* \cap Z_1) = r(B_0 \cup (I_M^* \cap Z_1)) - r(I_M^* \cap Z_1) \\
    \le r(B_0 \cup Z_1) - r(Z_1) + \epsilon |B_0| = r(B_0 \mid Z_1) + \epsilon |B_0| < 3\epsilon |B_0| \ . \qedhere
\end{multline*}
\end{proof}
Next we will prove that whenever the algorithm returns failure, there exists
such a certificate.

\begin{lemma}
    If $|B| < \epsilon |B_0|$ then there exists a certificate that
    proves infeasibility.
\end{lemma}
\begin{proof}
We set $Z_2 = A \cup B$ and $Z_1 = C \cup B$.
Then by Lemma~\ref{lem:addable4} we have
\begin{equation*}
    r(B_0 \mid Z_1) \le r(B_0 \mid C) < 2\epsilon |B_0| \ .
\end{equation*}
Further,
\begin{multline*}
    r(Z_1) \le r(C) + r(B)
    = |C| + \epsilon |B_0|
    \le |Z_1| + \epsilon |B_0| - 0.5 |B| - 0.5 |B| \\
    \le |Z_1| + \epsilon |B_0| - 0.5 |B| - (0.5 - \epsilon) |A| \le |C_1| - (0.5 - 2\epsilon) |Z_2| + 2\epsilon |B_0| \ .
\end{multline*}
Here we use that $|B| \ge (1 - 2\epsilon) |A|$, see Lemma~\ref{lem:size-B}.
It follows from Lemma~\ref{lem:large-subm-value} and Lemma~\ref{lem:decr-marg} that
$f(i \mid b \cdot (Z_2 - i)) \le f(i \mid b \cdot (A\cup B - i)) < 2b$ for each $i\in C_2 \setminus A_I$
and $|A_I| \le \epsilon |A| \le \epsilon |Z_2|$ since
otherwise we would have terminated successfully.

Finally, $f(i \mid 2b \cdot C_2) \le f(i \mid 2b \cdot A) < 2b$ for all $i\in C\setminus A$ follows immediately from the definition of $C$.
\end{proof}

\begin{lemma}
    When a recursive call returns failure, then there exists a certificate that
    proves infeasibility.
\end{lemma}
\begin{proof}
Assume that a recursive call has failed.
Let $Z'_1, Z'_2$ be the certificate returned by it.
We set $Z_1 = Z'_1 \cup C \cup B$ and $Z_2 = Z'_2 \cup A \cup B$.
Then by Lemma~\ref{lem:addable4},
\begin{equation*}
    r(B_0 \mid Z_1) \le r(B_0 \mid C) < 2 \epsilon |B_0| \ .
\end{equation*}
Further,
\begin{align*}
    r(Z_1) &= r(Z'_1 \cup C \cup B) \\
    &= r(Z'_1 \mid C) + r(C) + r(B \mid Z'_1 \cup C) \\
    &\le r'(Z'_1) + |C| + r'(B \mid Z'_1) \\
    &\le |Z'_1| - (0.5 - 2\epsilon) |Z'_2| + 2\epsilon |B| + |C| + \epsilon |B_0| \\
    &= |Z_1| - (0.5 - 2\epsilon) |B| - 0.5 |B| - (0.5 - 2\epsilon) |Z'_2| + \epsilon |B_0| \\
    &\le |Z_1| - (0.5 - 2\epsilon) |Z_2| + \epsilon |B_0| \ .
\end{align*}
Moreover, for all $i\in C_2\setminus A_I$
it holds that $f(i\mid b\cdot (Z_2 - i)) \le f(i\mid b\cdot (A\cup B - i)) < 2b$
Similarly, for a $1-\epsilon$ fraction of $Z'_2$ it
holds that $f(i\mid b \cdot (Z_2 - i)) \le f(i\mid b \cdot (Z'_2 - i)) < b$. Since $|A_I| < \epsilon |A| \le \epsilon |A\cup B|$ Property~\eqref{eq:margins}
is satisfied.
Likewise $f(i \mid 2b \cdot Z_2) \le f(i \mid 2b \cdot Z'_2) < 2b$ for all $i\in Z'_1 \setminus Z_2$
and $f(i \mid 2b \cdot Z_2) \le f(i \mid A) < 2b$ for all $i\in C\setminus Z_2$. Thus, also the last property holds.
\end{proof} 

\subsection{Running time}
We are going to bound only the number of nodes in the recursion tree. It is clear that the overhead of operations outside the recursive call is polynomially bounded. To this end, we focus on the sets $B_0$ and $B$.
Initially, the set $B$ created in the
algorithm will have size at most $n$. Then
with every recursive call it decreases by a factor of $(1 - \epsilon^2)$, see Lemma~\ref{lem:decr-B},
but never below $\epsilon |B_0|$ (else, the algorithm terminates immediately).

For a fixed instance,
let $T(k)$ be the maximum number of nodes in the recursion tree of the algorithm over all inputs $B_0, X, Y$
where $|B_0| \ge 1 / (1 - \epsilon^2)^k$. Then $T(k)$ is monotone decreasing, i.e., $T(k') \le T(k)$ for $k' \ge k$, simply because the set over which the maximum is taken is smaller. Let $\ell = \lceil\log_{1/(1 - \epsilon^2)^k}(n)\rceil$.
Then,
\begin{equation*}
    T(k) \le 1 + \sum_{i = k + 1}^{\ell} T(i) \quad \text{and} \quad T(\ell) = 1 \ .
\end{equation*}
Here, the sum starts at $k+1$, since the smallest size of $B'_0$ of the recursion satisfies $|B'_0| \ge (1 + \epsilon)|B_0| \ge |B_0| / (1 - \epsilon^2)$.
The numbers $T(\ell), T(\ell - 1), T(\ell - 2),\dotsc,T(1)$
are similar to the Fibonacci series (except for the
addition of $1$ for the current node) and it
can easily be shown by induction that
$T(k) \le 2^{\ell - k} \le 2^{\ell} \le n^{\cO_{\epsilon}(1)}$ for all $k$.
\end{document}